\documentclass[11pt,a4paper,twoside]{article}

\usepackage{amsmath,amssymb,amsthm}
\usepackage{microtype}

\usepackage[
    a4paper,
    top=27mm,
    bottom=27mm,
    inner=27mm,
    outer=27mm,
    footskip=12mm
]{geometry}

\usepackage{needspace}
\usepackage{etoolbox}
\usepackage{titlesec}

\titleformat{\section}
    {\Large\sffamily\bfseries}
    {\thesection}
    {0.65em}
    {}
\titleformat{\subsection}
    {\large\sffamily\bfseries}
    {\thesubsection}
    {0.65em}
    {}
\titleformat{\subsubsection}
    {\normalsize\sffamily\bfseries}
    {\thesubsubsection}
    {0.65em}
    {}

\titlespacing*{\section}
    {0pt}
    {3.2ex plus 0.8ex minus 0.4ex}
    {1.6ex plus 0.3ex}
\titlespacing*{\subsection}
    {0pt}
    {2.7ex plus 0.7ex minus 0.3ex}
    {1.2ex plus 0.2ex}
\titlespacing*{\subsubsection}
    {0pt}
    {2.3ex plus 0.6ex minus 0.3ex}
    {0.9ex plus 0.2ex}

\pretocmd{\section}{\Needspace{5\baselineskip}}{}{}
\pretocmd{\subsection}{\Needspace{4\baselineskip}}{}{}
\pretocmd{\subsubsection}{\Needspace{3\baselineskip}}{}{}

\numberwithin{equation}{section}

\usepackage{abstract}

\usepackage{graphicx}
\usepackage{caption}
\usepackage{enumitem}
\setlist{
    topsep=5pt,
    itemsep=2pt,
    parsep=0pt,
    partopsep=0pt
}
\setlist[itemize]{leftmargin=2em}
\setlist[enumerate]{leftmargin=2em}

\newtheoremstyle{thesisthm}
    {6pt}
    {6pt}
    {\normalfont}
    {}
    {\bfseries}
    {.}
    {0.5em}
    {}
\theoremstyle{thesisthm}
\newtheorem{thm}{Theorem}[section]
\newtheorem{lem}[thm]{Lemma}
\newtheorem{prop}[thm]{Proposition}

\usepackage{xcolor}
\definecolor{linkblue}{HTML}{4F6D8A}
\definecolor{citegreen}{HTML}{527A61}
\definecolor{urlred}{HTML}{8A5A63}
\usepackage[skins,breakable]{tcolorbox}
\tcbset{
    base/.style={
        enhanced,
        breakable,
        sharp corners,
        boxrule=0pt,
        leftrule=1mm,
        colframe=black!35,
        colback=black!3,
        colbacktitle=black!7,
        coltitle=black,
        fonttitle=\bfseries,
        left=2.5mm,
        right=2.5mm,
        top=1.5mm,
        bottom=1.5mm,
        toptitle=0.8mm,
        bottomtitle=0.8mm,
        before skip=8pt,
        after skip=8pt,
        title={#1}
    }
}
\newtcolorbox{subbox}[1]{base={#1}}

\usepackage{cite}
\usepackage[
    colorlinks=true,
    linkcolor=linkblue,
    citecolor=citegreen,
    urlcolor=urlred,
    filecolor=urlred,
    bookmarks=true,
    bookmarksopen=true,
    bookmarksnumbered=true,
    pagebackref=true
]{hyperref}

\renewcommand*{\backref}[1]{}
\renewcommand*{\backrefalt}[4]{%
    \ifcase #1
    \or
        \space{\footnotesize\textit{Cited on p.~#2.}}%
    \else
        \space{\footnotesize\textit{Cited on pp.~#2.}}%
    \fi
}

\begin{document}
\thispagestyle{plain}

\begin{flushright}
{\tt\normalsize ZMP-HH-26/34}\\
\end{flushright}

\vspace*{2.5cm}

\begin{center}
    {\LARGE\bfseries
        No Shortcuts to Infinity
        \par
    }

    \vspace{0.35cm}

    {\Large
        Weil--Petersson Distance in Calabi--Yau 
        Complex Structure Moduli Space
        \par
    }

    \vspace{1.6cm}

    {\large\bfseries Jeroen Monnee}
    \footnote{\href{mailto:jeroen.monnee@desy.de}
    {jeroen.monnee@desy.de}}

    \vspace{0.5cm}

    {\small
        II. Institut f\"ur Theoretische Physik,
        Universit\"at Hamburg,\\
        Notkestrasse 9, 22607 Hamburg, Germany
        \par
    }
\end{center}

\vspace{0.5cm}

\begin{abstract}
\noindent
Infinite-distance limits play a central role in the Swampland program, but understanding their behaviour at intersections of boundary divisors is subtle, as it requires control over simultaneous limits of several moduli. To address this, we establish a necessary and sufficient criterion for a boundary stratum of arbitrary codimension in the complex-structure moduli space of Calabi--Yau manifolds to lie at infinite Weil--Petersson distance. In particular, intersections of infinite-distance boundary divisors necessarily remain at infinite distance. The result holds for the full Weil--Petersson metric and arbitrary trajectories, including those with unrestricted hierarchies in the saxions and axion winding. In compactifications of Type IIB string theory on Calabi--Yau threefolds, our proof is based on an inequality relating the Weil--Petersson length of a trajectory to the BPS mass associated with an auxiliary charge vector. If this mass tends to zero while the charge-to-mass ratio remains uniformly bounded, every trajectory approaching the limit has infinite length. Using asymptotic Hodge theory, we construct an auxiliary charge vector satisfying these conditions and derive the criterion directly from the limiting mixed Hodge structure associated with the degeneration. More generally, our result applies to polarized variations of Hodge structure of Calabi--Yau type. In particular, we prove Wang's conjecture on finite Weil--Petersson distance for degenerations over higher-dimensional bases.
\end{abstract}

\newpage

\hypersetup{linkcolor=black}
\tableofcontents
\hypersetup{linkcolor=linkblue}

\setcounter{footnote}{0}
\newpage

\setcounter{section}{0}

\section{Introduction}
The Swampland Distance Conjecture (SDC) \cite{Vafa:2005ui,Ooguri:2006in} makes a prediction for the behaviour of effective theories that are consistent with quantum gravity, in asymptotic regions of their moduli space. More precisely, it states that, as the geodesic distance from a fixed point in the moduli space becomes arbitrarily large, an infinite tower of states becomes exponentially light in Planck units. Such limits therefore signal the breakdown of an effective description with a fixed cutoff. Over the past years, substantial progress on understanding the SDC has been made through both top-down tests in explicit string compactifications and bottom-up arguments based on general principles of quantum gravity. Testing the conjecture requires detailed control over the existence and stability of towers of states, their asymptotic mass scales, and the relation between these scales and the geodesic distance. In string compactifications, these questions are closely tied to the geometry of the internal space, leading to a rich interplay between geometry and fundamental properties of quantum gravity. We refer the reader to \cite{
Grimm:2018ohb,
Heidenreich:2018kpg,
Grimm:2018cpv,
Blumenhagen:2018nts,
Lee:2019oct,
Corvilain:2018lgw,
Joshi:2019nzi,
Marchesano:2019ifh,
Baume:2019sry,
Gendler:2020dfp,
Bastian:2020egp,
Calderon-Infante:2020dhm,
Lanza:2021udy,
Etheredge:2022opl,
Grimm:2022sbl,
Hassfeld:2025uoy,
Baines:2025upi,
Monnee:2025ynn} for a selection of works relevant for the present context, as well as \cite{Brennan:2017rbf,Palti:2019pca,vanBeest:2021lhn,Grana:2021zvf,Agmon:2022thq} for reviews and further references.

In this work, we address a more basic question which underlies these analyses: 
\begin{quotation}
    \textit{When} does a point in the moduli space lie at infinite distance?
\end{quotation}
Answering this question is crucial to identify the limits to which the SDC applies, and is therefore an essential part of understanding the domain of validity of effective theories coupled to gravity. 

This question becomes particularly delicate when several moduli are taken to be large simultaneously. As a basic example, consider the intersection $\Delta_{12}$ of two infinite-distance divisors $\Delta_1$ and $\Delta_2$, see also Figure~\ref{fig:boundary-intersection}. One can then ask the concrete question: Must their intersection also lie at infinite distance, or could there exist some exotic trajectory (a ``shortcut to infinity'') along which the distance to $\Delta_{12}$ is actually finite? This is complicated, because to argue that $\Delta_{12}$ is at infinite distance, one must consider every possible trajectory approaching it, including for example those that repeatedly bounce between asymptotic regimes and thus have no clear hierarchy in the moduli, or those that feature extreme winding in axionic directions. Furthermore, if the moduli space metric is only assumed to be K\"ahler (without additional structure), there are simple examples in which such ``shortcuts to infinity'' do exist, as we illustrate explicitly in Section~\ref{subsec:limits}.

\begin{figure}[h!]
    \centering
    \includegraphics[width=0.55\linewidth]{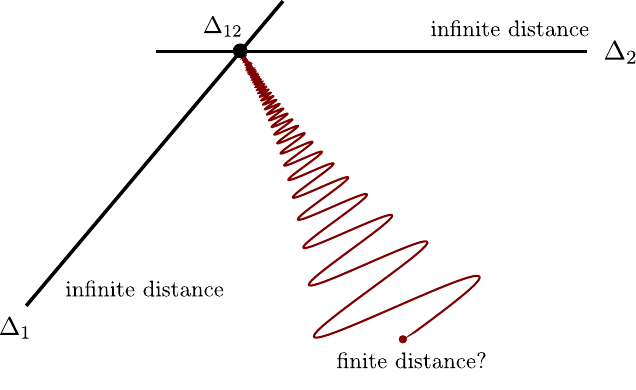}
    \caption{Two infinite-distance limits $\Delta_1$ and $\Delta_2$, and their intersection $\Delta_{12}$. We are interested in understanding whether the intersection is always at infinite distance, or whether there could exist some exotic trajectory along which the distance to $\Delta_{12}$ is finite.}
    \label{fig:boundary-intersection}
\end{figure}
Nevertheless, we expect the additional structure of moduli spaces arising in consistent theories of quantum gravity to exclude such finite-length shortcuts. A physical motivation for this expectation comes from the compatibility of the dual descriptions associated with different asymptotic regimes. Constraints on how such regimes fit together have been studied using the scalar charge-to-mass vectors of light towers \cite{Etheredge:2024tok}, while broader restrictions on moduli spaces and their duality groups have been motivated by finiteness arguments \cite{Hamada:2021yxy,Delgado:2024skw}. Indeed, for the intersection of infinite-distance boundaries to be at finite distance would require the distance estimates valid near the individual divisors to lose uniform control, which appears difficult to reconcile with a consistent matching of duality frames, including their light spectra and couplings.\footnote{See also \cite{Kaufmann_to_appear} for a related discussion.} This motivates investigating the question in moduli spaces that arise in string compactifications, where the underlying geometry provides additional tools to establish such restrictions. 

Calabi--Yau compactifications provide a particularly fruitful setting in which to study these questions. More precisely, we will consider compactifications of Type IIB string theory on Calabi--Yau threefolds, without fluxes or orientifolds. We will focus our attention on the vector multiplet moduli space of the resulting four-dimensional $\mathcal{N}=2$ supergravity theory, which is identified with the complex-structure moduli space of the Calabi--Yau threefold, equipped with the Weil--Petersson metric. In this setting, asymptotic Hodge theory \cite{schmid,CKS} has been used to classify degeneration limits into types I--IV and construct candidate towers of asymptotically massless states \cite{Grimm:2018ohb,Grimm:2018cpv,Corvilain:2018lgw}, and to determine the asymptotic behaviour of various physical couplings \cite{Grimm:2019wtx,Gendler:2020dfp,Bastian:2020egp,Bastian:2021hpc,Palti:2021ubp,Grimm:2022xmj}, periods \cite{Bastian:2021eom,Bastian:2023shf}, and the Weil--Petersson metric itself \cite{Grimm:2020cda,Grimm:2021ikg}, see also \cite{vandeHeisteeg:2022gsp,Monnee:2024gsq} for reviews and additional references. Furthermore, the physical interpretation of these limits in light of the Emergent String Conjecture \cite{Lee:2019oct} was recently illuminated in \cite{Hassfeld:2025uoy,Monnee:2025ynn,Monnee:2025msf}. Related geometric analyses have connected infinite-distance limits to emergent string limits and decompactification limits in the complex-structure moduli space of F-theory compactifications on elliptic K3 surfaces \cite{Lee:2021qkx,Lee:2021usk} and elliptic Calabi--Yau threefolds \cite{Alvarez-Garcia:2023gdd,Alvarez-Garcia:2023qqj}. 

On the mathematical side, a Hodge-theoretic criterion for finite Weil--Petersson distance in one-parameter degenerations was first established by Wang \cite{Wang1997}, and subsequently further investigated by Lee \cite{Lee2018} for higher-dimensional moduli spaces. In particular, it was shown that one-parameter degenerations are at infinite distance if and only if the singularity type is either II, III, or IV. We will discuss these results and their relation to our theorem in more detail below. The typical strategy in both these works is to write down an asymptotic expansion for the K\"ahler potential, and subsequently address the effect of corrections. However, this becomes very delicate in higher-dimensional moduli spaces. For example, the leading approximation to the K\"ahler potential need not give a non-degenerate approximation to the metric in every direction due to the presence of so-called metric-essential instantons \cite{Bastian:2021eom,Bastian:2021hpc}, and taking their effect into account appears difficult in general. 

\subsection*{Main results}
Instead of analysing the metric directly, our strategy is to bound the length of any trajectory by the variation of a \textit{scalar} quantity.\footnote{For related uses of logarithmic mass gradients in the study of the Distance Conjecture, see
\cite{Etheredge:2023zjk}.} In Type IIB string theory compactified on a Calabi--Yau threefold $Y$, this quantity is precisely the BPS mass associated with an auxiliary charge vector $q$. In the notation of Section~\ref{sec:physical-criterion}, we define
\begin{equation}
    \mathcal{M}_q=\frac{|(q,\Omega)|}{\|\Omega\|}\,,
    \qquad
    \mathcal{Q}_q=\frac{1}{\sqrt{2}}\|q\|\,,
    \label{eq:intro-mass-charge}
\end{equation}
where $\Omega$ is the holomorphic three-form, $(\cdot,\cdot)$ is the intersection pairing, and $\|\cdot\|$ denotes the Hodge norm. For a physical BPS state, these functions give its mass in Planck units and its physical charge with respect to the $\mathrm{U}(1)$ gauge fields in the gravity and vector multiplets. The relation between charge-to-mass ratios, scalar couplings, and asymptotic mass scales has already played an important role in connecting the Weak Gravity Conjecture and the Distance Conjecture \cite{Palti:2017elp,Gendler:2020dfp,Bastian:2020egp}. Our focus here is to use these quantities to formulate the following criterion for infinite distance. 
\vspace{0.3cm}
\begin{subbox}{Infinite distance in Calabi--Yau moduli space}
    Consider Type IIB string theory compactified on a Calabi--Yau threefold $Y$. Let $p_\infty$ be a point on the boundary of its complex-structure moduli space. Then every trajectory approaching $p_\infty$ has infinite length with respect to the full Weil--Petersson metric if and only if the degeneration at $p_\infty$ is of type II, III, or IV. \\

    More precisely, if the degeneration is of type II, III, or IV, there exists a non-zero, monodromy-invariant auxiliary charge $q\in H^3(Y,\mathbb{C})$ and a constant $C>0$ such that, throughout a sufficiently small neighbourhood of $p_\infty$, the following holds:
    \begin{align}
    \label{eq:intro-charge-to-mass}&\text{(1)}:\qquad \frac{\mathcal{Q}_q}{\mathcal{M}_q} \leq C\,,\\
    \label{eq:intro-mass}&\text{(2)}:\qquad \mathcal{M}_q\to 0\,,\qquad\text{as }p_\infty\text{ is approached}.
    \end{align}
    In addition, for every trajectory $\gamma$ connecting two points $p_1$ and $p_2$, its Weil--Petersson length satisfies
    \begin{equation}
        L(\gamma)
        \geq\frac{1}{C}
        \left|\log\frac{\mathcal{M}_q(p_2)}{\mathcal{M}_q(p_1)}\right|\,.
        \label{eq:intro-length-bound}
    \end{equation}
    Consequently, every trajectory approaching $p_\infty$ has infinite length with respect to the full Weil--Petersson metric.
\end{subbox}
\vspace{0.3cm}
We emphasize that the bound is uniform in a neighbourhood of the boundary point. In particular, it also applies to trajectories which feature arbitrary hierarchies in the saxions, as well as extreme winding in the axionic directions. Since the enhancement rules of \cite{Kerr2017,Grimm:2018cpv,Corvilain:2018lgw} forbid the singularity type to decrease, this implies that the intersection of two infinite-distance divisors is always at infinite distance, and similarly for higher codimension intersections. 

The proof itself only uses the period and Hodge norm in \eqref{eq:intro-mass-charge}. In particular, the auxiliary charge $q$ may be complex, and neither the existence nor the stability of a corresponding BPS state is necessary.\footnote{With a bit more work, the charge $q$ can be chosen to be real, and in fact integral if the underlying variation of Hodge structure is integral.} In fact, the same argument applies more generally to effective polarized real variations of Hodge structure of weight $D\geq1$ with $h^{D,0}=1$, which we call variations of Calabi--Yau $D$-fold type. This includes the geometric setting of families of polarized Calabi--Yau manifolds, but does not require such a geometric realization. 

To state the general result, we work locally over $(\Delta^*)^n\times\Delta^m$ near a boundary point $p_\infty$ where $n$ boundary divisors intersect normally. Here $\Delta$ is a disk, $\Delta^*$ is a punctured disk, and the coordinates $z^i$ and $\zeta^\alpha$ describe the degenerating and spectator directions, respectively, and are chosen such that $p_\infty = (0,0)$. We assume that the local monodromies around the punctured disks are quasi-unipotent.\footnote{In the geometric setting, this is guaranteed by the Monodromy Theorem \cite{Landman1973,schmid}.} After a finite base change, their logarithms $N_i=\log T_i$ are commuting nilpotent operators. The nilpotent orbit theorem of Schmid then gives a local description of the holomorphic generator $\Omega$ of the Hodge filtration $F^D$ of the form \cite{schmid,CKS}
\begin{equation}
    \Omega(z,\zeta)
    =\exp\left[\sum_{i=1}^n \frac{\log z^i}{2\pi\mathrm{i}}N_i\right]A(z,\zeta)\,,\qquad A(0,0) = a_0\,,
    \label{eq:intro-nilpotent-orbit}
\end{equation}
where $A$ extends holomorphically to the boundary. Introduce the notation $N=\sum_{i=1}^nN_i$. We prove the following theorem.
\vspace{0.3cm}
\begin{subbox}{\begin{thm}\label{thm:main}
    A Hodge-theoretic criterion for infinite distance
\end{thm}}
    Under the assumptions above
    \begin{equation}
        p_\infty\text{ lies at infinite distance}
        \quad\Longleftrightarrow\quad
        Na_0\neq0\,.
        \label{eq:intro-distance-criterion}
    \end{equation}
\end{subbox}
\vspace{0.3cm}
For one-parameter Calabi--Yau degenerations, the equivalence between finite distance and $Na_0=0$ was established by Wang \cite[Theorem~1.1]{Wang1997}. Its extension to higher-dimensional moduli spaces is Wang's conjecture, as stated in \cite[Conjecture~0.1]{Lee2018}. The implication $\implies$ follows by reduction to the one-parameter case, as shown by Lee \cite[Theorem~2.2]{Lee2018}. The reverse implication $\impliedby$ requires control of every trajectory approaching $p_\infty$ and is therefore substantially more difficult. For two-parameter degenerations of Calabi--Yau threefolds, Lee proved infinite distance at intersections of two infinite divisors for the metric induced by the dominant terms of candidate Weil--Petersson K\"ahler potentials \cite[Proposition~3.5]{Lee2018}. They also established infinite distance in the full metric along constant-axion slices at intersections of one finite and one infinite-distance divisor \cite[Theorem~3.3]{Lee2018}. Our result establishes the implication $\impliedby$ for the full metric, arbitrary trajectories, and intersections of arbitrary codimension. In particular, it proves Wang's conjecture in the setting stated above.

The central step in the proof is the construction of a single auxiliary charge
\begin{equation}
    q=N^d\overline{a_0}\,,
    \qquad
    d=\max\{r\geq0:N^ra_0\neq0\}\,,
    \label{eq:intro-auxiliary-charge}
\end{equation}
satisfying the conditions \eqref{eq:intro-charge-to-mass} and \eqref{eq:intro-mass}. We show that $q$ is invariant under every local monodromy and has the weight-filtration properties needed to control its Hodge norm in every growth sector, using the growth theorem of Cattani, Kaplan, and Schmid \cite{CKS}. Importantly, the same charge works in every growth sector. Since there are only finitely many sectors, their combined estimates yield a uniform bound throughout a neighbourhood of the boundary point.\footnote{A similar strategy was also employed in \cite{Grimm:2022sbl,Grimm:2023lrf}.} This gives a bound for the full metric without requiring an asymptotic expansion of its individual components.

\subsubsection*{Overview of the paper}
The paper is organized as follows. In Section~\ref{sec:physical-criterion}, we review the relevant aspects of Type IIB compactifications on Calabi--Yau threefolds, explain the difficulties that arise at boundary intersections, and prove a physical criterion for infinite distance in Proposition~\ref{prop:distance_criterion}. In Section~\ref{sec:Hodge}, we introduce the necessary Hodge-theoretic setting and prove Theorem~\ref{thm:main} by constructing the auxiliary charge and establishing uniform estimates on the mass and charge-to-mass ratio. We conclude with a discussion in Section~\ref{sec:discussion}. Appendix~\ref{app-lemma} provides additional background on Deligne splittings and the associated $\mathfrak{sl}(2)$-algebras, and contains the proof of Lemma~\ref{lem:weights}. 

\section{A physical criterion for infinite distance}
\label{sec:physical-criterion}
In this section, we formulate a sufficient criterion for infinite distance in terms of the mass and physical charge of an auxiliary charge vector. To this end, we first recall in Section~\ref{subsec:sugra} the relevant aspects of four-dimensional $\mathcal{N}=2$ supergravity obtained from Type IIB compactifications on Calabi--Yau threefolds. Readers familiar with this material can safely skip this section. In Section~\ref{subsec:limits}, we then discuss the distinction between one-parameter limits and limits towards intersections of boundary divisors, and illustrate with a simple K\"ahler example why the latter require additional care. Furthermore, we discuss the difficulties that a strategy based on an asymptotic analysis of the metric faces. In Section~\ref{subsec:sugra-criterion}, we address these issues by employing an alternative strategy, and present our main criterion in Proposition~\ref{prop:distance_criterion}. In short, we show that if the mass of an auxiliary charge vector tends to zero while its charge-to-mass ratio remains bounded, then every trajectory approaching the limit has infinite length with respect to the moduli space metric. Although initially motivated in the language of supergravity, the argument can be formulated entirely in terms of Hodge norms and periods. This will allow us to apply the criterion to general effective variations of Hodge structure of Calabi--Yau $D$-fold type in Section~\ref{sec:Hodge}. Finally, in Section~\ref{subsec:SDC}, we make some additional comments concerning the Swampland Distance Conjecture. 

\subsection{4d \texorpdfstring{$\mathcal{N}=2$}{N=2} from Type IIB on Calabi--Yau threefolds}
\label{subsec:sugra}
We consider the low-energy effective description of Type IIB string theory compactified on a family of Calabi--Yau threefolds $Y$, varying in moduli, without fluxes or an orientifold projection. The resulting theory is a four-dimensional $\mathcal{N}=2$ supergravity theory coupled to $h^{2,1}(Y)$ vector multiplets and $h^{1,1}(Y)+1$ hypermultiplets. The moduli space $\mathcal{M}_{\mathrm{mod}}$ parametrized by the vacuum expectation values of the massless scalars in the theory locally factorizes into a direct product
\begin{equation}
\label{eq:moduli_space_full}
    \mathcal{M}_{\mathrm{mod}} = \mathcal{M}_V\times \mathcal{M}_H\,,
\end{equation}
where $\mathcal{M}_V$ denotes the vector multiplet moduli space, which is a projective special K\"ahler manifold of complex dimension $h^{2,1}(Y)$, and $\mathcal{M}_H$ denotes the hypermultiplet moduli space, which is a quaternionic-K\"ahler manifold of real dimension $4h^{1,1}+4$. Geometrically, the former parametrizes the complex structure deformations of the Calabi--Yau threefold $Y$, while the latter incorporates the (complexified) K\"ahler deformations of $Y$ as well as the axio-dilaton. Together with the factorization \eqref{eq:moduli_space_full}, this implies that the vector multiplet moduli space in Type IIB compactifications on Calabi--Yau threefolds is classically exact, i.e.~it is protected from perturbative and non-perturbative corrections in $\alpha'$ and $g_s$. Throughout this section, we restrict our attention to the vector multiplet moduli space, keeping the hypermultiplet scalars fixed, and write $\mathcal{M}=\mathcal{M}_V$, for simplicity. 

\subsubsection*{Action and distance}
At the two-derivative level, the bosonic part of the low-energy effective action describing the gravity multiplet and the vector multiplets reads
\begin{equation}\label{eq:action}
    S = \int\left(\frac{1}{2}M_{\mathrm{pl}}^2R\star 1-M_{\mathrm{pl}}^2G_{i\bar{\jmath}}\,\mathrm{d}z^i\wedge\star\,\mathrm{d}\bar{z}^{\bar{\jmath}}+\frac{1}{4}\mathcal{I}_{IJ}\,F^I\wedge\star\,F^J+\frac{1}{4}\mathcal{R}_{IJ}\,F^I\wedge F^J\right)\,.
\end{equation}
Here $z^i$, with $i=1,\ldots, h^{2,1}(Y)$, denote the complex structure moduli of $Y$, and play a central role in this work. Furthermore, $F^I=\mathrm{d}A^I$, with $I=0,\ldots, h^{2,1}(Y)$, denote the field strengths of the $\mathrm{U}(1)$ gauge fields residing in the vector multiplets and the gravity multiplet. The matrices $\mathcal{R}_{IJ}$ and $\mathcal{I}_{IJ}$ are the real and imaginary parts of the gauge kinetic matrix. For further details on the structure of four-dimensional $\mathcal{N}=2$ supergravity we refer the reader to~\cite{Andrianopoli:1996cm}. We follow the conventions of \cite{Bastian:2020egp}. 

The kinetic term for $z^i$ defines the physical metric $G_{i\bar{\jmath}}$ on the moduli space, also known as the Weil--Petersson metric. Locally, it can be expressed in terms of a K\"ahler potential $K$ as 
\begin{equation}\label{eq:WP-metric}
    ds^2 = 2G_{i\bar{\jmath}}\, \mathrm{d}z^i \mathrm{d}\bar{z}^{\bar{\jmath}}\,,\quad G_{i\bar{\jmath}} = \partial_i \bar{\partial}_{\bar{\jmath}}K\,,\qquad K = -\log \left[\mathrm{i} \int \Omega\wedge\bar{\Omega} \right]\,,
\end{equation}
where $\Omega$ denotes the holomorphic three-form on $Y$. Given a trajectory $\gamma:\lambda\mapsto z^i(\lambda)$ connecting two points $z^i(\lambda_1)$ and $z^i(\lambda_2)$, its length $L(\gamma)$ with respect to the Weil--Petersson metric \eqref{eq:WP-metric} is given by
\begin{equation}\label{eq:WP_length}
    L(\gamma) = \int_{\lambda_1}^{\lambda_2} \sqrt{2G_{i\bar{\jmath}}\frac{\mathrm{d}z^i}{\mathrm{d}\lambda} \frac{\mathrm{d}\bar{z}^{\bar{\jmath}}}{\mathrm{d}\lambda}} d\lambda\,.
\end{equation}

\subsubsection*{BPS particles}
A natural way to obtain particle states in this theory is by considering D3-branes wrapped on internal homology three-cycles. When the wrapped three-cycle is special Lagrangian, the resulting particle states are BPS.\footnote{A real three-dimensional submanifold $\Gamma\subset Y$ is Lagrangian if $J|_\Gamma=0$, where $J$ denotes the K\"ahler form on $Y$. It is special Lagrangian if, for some constant phase $\theta$, it is calibrated by $\operatorname{Re}(e^{-i\theta}\Omega)$, with $\Omega$ suitably normalized. In particular, $\operatorname{Im}(e^{-i\theta}\Omega)|_\Gamma=0$. A compact calibrated representative minimizes volume in its homology class.} Of particular importance for us are the mass and physical charge of such particle states.

Let $\Gamma\in H_3(Y,\mathbb{Z})$ be an integral homology class with a special Lagrangian representative, and consider a BPS particle that arises from wrapping a D3-brane on this cycle. For later purposes, it will be convenient to label the particle by the integral cohomology class $q \in H^{3}(Y,\mathbb{Z})$ that is Poincar\'e dual to $\Gamma$. Its mass is given by (see for example \cite{Palti:2017elp})
\begin{equation}
\label{eq:mass_central-charge}
    \frac{\mathcal{M}_q}{M_{\mathrm{pl}}} = |\mathcal{Z}_q| = e^{K/2}\left| \int_\Gamma \Omega \right| = e^{K/2}\left|\int_Y q\wedge \Omega \right|\,,
\end{equation}
where $\mathcal{Z}_q$ denotes the central charge. 
The mass is measured in units of the four-dimensional Planck mass $M_{\mathrm{pl}}$, whose relation to the string scale $M_s$ depends on the string coupling as well as the volume of $Y$, both of which appear as scalar fields in the hypermultiplet sector of the moduli space.\footnote{More precisely, $M_{\mathrm{pl}}$ is related to the string scale $M_s$ as
\begin{equation*}
    M_s \sim \frac{g_s}{\sqrt{\mathcal{V}_s}}M_{\mathrm{pl}}\,,
\end{equation*}
where $g_s$ denotes the string coupling and $\mathcal{V}_s$ denotes the volume of the Calabi--Yau threefold in units of the string length.} Since our analysis pertains purely to the vector multiplet sector, we may take $M_{\mathrm{pl}}$ as constant and will henceforth set $M_{\mathrm{pl}}=1$ without loss of generality. In particular, all the moduli-dependence of the mass then comes from the periods of the holomorphic three-form $\Omega$ and the factor $e^{K/2}$.\footnote{It should be noted that $q$ itself is a fixed topological quantity, meaning that it is fixed under parallel transport by the Gauss--Manin connection. Nevertheless, the existence and stability of a BPS state with charge $q$ are separate questions that can depend on the moduli; see, for example, \cite{Denef:2000nb,Palti:2021ubp}.}

Upon choosing an integral basis $\gamma^\alpha\in H^{3}(Y,\mathbb{Z})$ one can express $q=q_\alpha \gamma^\alpha$ as a vector with integral entries $q_\alpha\in\mathbb{Z}$, for $\alpha=1,\ldots, 2+2h^{2,1}$. Geometrically, these integers specify the D3-brane wrapping numbers with respect to the corresponding basis of dual three-cycles. Physically, they correspond to the electric and magnetic U(1) charges of the particle under the $h^{2,1}(Y)$ gauge fields residing in the vector multiplets, together with the graviphoton residing in the gravity multiplet. For this reason, we will often refer to $q$ as a ``charge vector''. The physical charge $\mathcal{Q}_q$ that measures the strength of the gauge interactions of the particle with charge $q$ is given by  (see for example \cite{Palti:2017elp})
\begin{equation}
    \mathcal{Q}_q^2 = \frac{1}{2}\int_Y q\wedge\star\,q\,.
\end{equation}

\subsection{Infinite distance limits in complex structure moduli space}
\label{subsec:limits}
The central question of this work is the following.
\begin{subbox}{Question: }
    When does a point $p_\infty$ on the boundary
\begin{equation}
    \partial\mathcal{M}
    =\overline{\mathcal{M}}\setminus\mathcal{M}
\end{equation}
of a compactification of $\mathcal{M}$ lie at infinite distance with respect to the Weil--Petersson metric?
\end{subbox}
Here, by ``infinite distance'' we mean that every piecewise smooth trajectory approaching $p_\infty$ has infinite length, as measured by \eqref{eq:WP_length}. Conversely, if there exists a trajectory approaching $p_\infty$ with finite length, then $p_\infty$ lies at finite distance. As emphasized in the Introduction, this distinction is central to the Swampland program: the Distance Conjecture associates infinite-distance limits with an infinite tower of states becoming light, whereas finite-distance limits exhibit qualitatively very different physics. 

We work near a normal-crossing boundary point and choose local coordinates $z^1,\ldots, z^n$ such that the boundary components of $\partial\mathcal{M}$ are given by divisors $\Delta_i = \{z^i=0\}$.\footnote{Quasi-projectivity of the moduli space of polarized Calabi--Yau manifolds \cite{Viehweg},
together with a resolution of singularities whose existence is guaranteed by \cite{Hironaka},
allows us to always reduce to this simple normal-crossing setting.} Importantly, we are interested in cases where the boundary point $p_\infty$ is approached by sending multiple complex structure moduli $z^i$ to the boundary. To this end, we introduce the notation
\begin{equation}
    \Delta_I = \bigcap_{i\in I} \Delta_i\,,\qquad 
    \Delta_I^\circ
    =
    \Delta_I\setminus\bigcup_{j\notin I}\Delta_j\,,
\end{equation}
for the codimension-$|I|$ strata of $\partial\mathcal{M}$. To be precise, we use the notation $\Delta_I^\circ$ to denote the part of $\Delta_I$ that is away from an intersection with another boundary divisor. 

\subsubsection{Codimension 1}
\label{subsubsec:codim-1}

Let us first recall the well-understood case of a one-parameter limit. Here, the moduli space may have arbitrary dimension, but we approach a single boundary divisor \(\Delta^\circ=\{z=0\}\), away from its intersections with other boundary components. We denote the coordinates along \(\Delta\) by \(\zeta^\alpha\) and keep them fixed. For the transverse coordinate, we introduce
\begin{equation}
t=\frac{1}{2\pi \mathrm{i}}\log z=a+\mathrm{i}s\,,
\qquad
s=-\frac{1}{2\pi}\log|z|\,,
\end{equation}
such that the boundary is reached as \(s\to\infty\).

In a suitable Kähler gauge, the leading behaviour of the Kähler potential takes the form \cite{Grimm:2018ohb}
\[
K=-\log\left[c(\zeta,\bar\zeta)\,s^d+\cdots\right]\,,
\qquad c(\zeta,\bar\zeta)>0\,.
\]
The omitted terms consist of lower powers of $s$ and exponentially suppressed corrections, which may also depend on the axion $a$. For Calabi--Yau threefolds, the degree $d$ can take the values $0,1,2,3$, corresponding respectively to degenerations of type I, II, III, and IV \cite{Grimm:2018ohb,Grimm:2018cpv}. 

For $d>0$, restricting to a transverse slice with $\zeta^\alpha$ fixed gives the leading metric
\begin{equation}
    ds_{\mathrm{lead}}^2
    =\frac{d}{2s^2}\left(\mathrm{d}a^2+\mathrm{d}s^2\right).
\end{equation}
This is the standard Poincar\'e metric on the upper half-plane. In particular, every path on this slice satisfies
\begin{align}
    L_{\mathrm{lead}}(\gamma) & \geq \sqrt{\frac{d}{2}} \int_\gamma \frac{1}{s}\left| \frac{\mathrm{d}s}{\mathrm{d}\lambda}\right| d\lambda\,,\\ 
    &\geq \sqrt{\frac{d}{2}} \left| \int_\gamma \frac{1}{s}\frac{\mathrm{d}s}{\mathrm{d}\lambda} d\lambda \right|\,,\\
    &=
    \sqrt \frac{d}{2}
    \left|
        \log\frac{s(\lambda_2)}{s(\lambda_1)}
    \right|.
\end{align}
Thus, the length of any trajectory approaching $s\to\infty$, while keeping $\zeta^\alpha$ fixed, is infinite with respect to the leading metric. However, we emphasize that, from the above computation alone, we cannot yet conclude that $p_\infty$ lies at infinite distance. This requires taking into account the effects of both (a) the sub-leading terms in the metric, and (b) the terms in the metric involving the additional coordinates $\zeta^\alpha$. For codimension-one limits, the full analysis taking into account these effects was performed by Lee in \cite[Theorem~2.6]{Lee2018}. The upshot is that, for any trajectory $\gamma$ approaching $s\to\infty$, one has the bound
\begin{equation}
    L(\gamma)\geq \sqrt{\frac{d-\epsilon (h^{2,1}-1)}{2}} 
    \left|
        \log\frac{s(\lambda_2)}{s(\lambda_1)}
    \right|\pm \text{finite terms}\,,
\end{equation}
for arbitrarily small $\epsilon$. Thus, for $d>0$ one indeed finds that $L(\gamma)\to \infty$ as $s\to \infty$. In contrast, for $d=0$ it was shown by Wang that the boundary is at \textit{finite} distance \cite[Theorem~1.1]{Wang1997}.

\subsubsection{Codimension $>1$}
The preceding discussion illustrates why one-parameter limits are comparatively tractable: only one variable becomes large, and $\log s$ provides a natural quantity that controls the distance. The story is quite different in codimension greater than 1, i.e. for limits that involve multiple saxions becoming large. In this case, their relative rates of approach become an essential part of the problem. For example, one may raise the following question. 

\begin{subbox}{Question: Intersections of infinite-distance divisors}
    If $\Delta_1^\circ$ and $\Delta_2^\circ$ are at infinite distance, must their intersection $\Delta_{12}^\circ$ also be at infinite distance?
\end{subbox}
To illustrate that this is a non-trivial question, we now discuss an example of a K\"ahler metric for which the answer to the above question is ``no''. The purpose of this example is to emphasize that the K\"ahler property of the metric alone does not suffice to determine the distance to the intersection. In our setting, the additional control will come from the underlying variation of polarized Hodge structure. 

\subsubsection*{A K\"ahler example}
Consider the positive-definite metric
\begin{equation}\label{eq:metric-example}
    ds^2
    =
    \frac{
        \left|
            \frac{\mathrm dz_1}{z_1}
            -\frac{\mathrm dz_2}{z_2}
        \right|^2
    }{1+x^2}
    +\left|z_2\,\mathrm dz_1+z_1\,\mathrm dz_2\right|^2\,,
    \qquad
    x=\log\left|\frac{z_1}{z_2}\right|^2\,,
\end{equation}
on the punctured bidisk $(\Delta^*)^2=
    \{(z_1,z_2)\in\mathbb C^2:
    0<|z_1|,|z_2|<1\}$. It is K\"ahler, with globally defined K\"ahler potential
\begin{equation}
    K=\frac{1}{2}\left[x\arctan x-\frac12\log(1+x^2)+|z_1z_2|^2\right]\,.
\end{equation}
Along a fixed transverse slice $z_2=\text{constant}$, the metric has asymptotics
\begin{equation}
    ds^2\sim
    \frac{|\mathrm dz_1|^2}
    {|z_1|^2\bigl(\log|z_1|^2\bigr)^2}
    \qquad\text{as }z_1\to0\,,
\end{equation}
and similarly with $z_1$ and $z_2$ interchanged. Thus, each transverse slice exhibits the familiar Poincar\'e cusp behaviour that we also saw in the one-parameter limits discussed earlier.

In fact, every path approaching either divisor $\Delta_1^\circ$ or $\Delta_2^\circ$ has infinite length. To see this, note that along an arbitrary trajectory,
\begin{equation}
    \dot x
    =
    2\operatorname{Re}\left(
        \frac{\dot z_1}{z_1}
        -\frac{\dot z_2}{z_2}
    \right).
\end{equation}
Dropping the nonnegative second term in the metric \eqref{eq:metric-example} therefore gives
\begin{equation}
    \frac{\mathrm ds}{\mathrm d\lambda}
    \geq
    \frac{
        \left|
            \frac{\dot z_1}{z_1}
            -\frac{\dot z_2}{z_2}
        \right|
    }{\sqrt{1+x^2}}
    \geq
    \frac{|\dot x|}{2\sqrt{1+x^2}}
    =
    \frac12
    \left|
        \frac{\mathrm d}{\mathrm d\lambda}
        \operatorname{arsinh}x
    \right|.
\end{equation}
Upon integration, we obtain
\begin{equation}
    L(\gamma)
    \geq
    \frac12
    \left|
        \operatorname{arsinh}x(\lambda_2)
        -
        \operatorname{arsinh}x(\lambda_1)
    \right|.
\end{equation}
Approaching $\Delta_1^\circ$ sends $x\to-\infty$, whereas approaching $\Delta_2^\circ$ sends $x\to+\infty$. Both limits therefore have infinite length, independently of the trajectory.

In contrast, we now consider a trajectory which approaches the intersection $\Delta_{12}$ along the diagonal $z_1=z_2=z$. Notably, in this case the first term in the metric \eqref{eq:metric-example} vanishes identically, such that 
\begin{equation}
    ds^2=4|z|^2\,\mathrm dz\,\mathrm d\bar z\,.
\end{equation}
Along the radial path $z=r e^{i\theta_0}$, with $0<r\leq r_0<1$, the distance to the origin is
\begin{equation}
    L(\gamma)=\int_0^{r_0}2r\,\mathrm dr=r_0^2<\infty\,.
\end{equation}
Hence the intersection $\Delta_{12}$ is at finite distance, even though each divisor $\Delta_i^\circ$ is at infinite distance.

\subsubsection*{Issues with an asymptotic analysis}
A natural strategy to address the question raised at the beginning of the previous subsection is to determine the asymptotic form of the Weil--Petersson metric and use it to estimate the length of an arbitrary trajectory. Near the boundary of the moduli space, the K\"ahler potential can be written schematically as (we suppress the dependence on the spectator moduli $\zeta^\alpha$)
\begin{equation}
    K=-\log\bigl[P(s)+E(a,s)\bigr]\,,\qquad \frac{\log z^i}{2\pi i} = t^i = a^i+\mathrm{i} s^i\,.
\end{equation}
where $P(s)$ is a polynomial in the saxions $s^i$ and $E(a,s)$ contains exponentially suppressed terms, possibly multiplied by polynomials in the $s^i$. This decomposition follows from the nilpotent orbit expansion, see for example \cite{Bastian:2021hpc}.

This suggests first studying the approximate metric
\begin{equation}
    G^{(0)}_{i\bar{\jmath}}
    =-\partial_i\partial_{\bar{\jmath}}\log P
\end{equation}
and then controlling the effect of the omitted terms from $E(a,s)$. Such an approach was pursued by Lee \cite{Lee2018}, classifying the dominant candidate potentials for two-parameter limits and establishing that the intersection of two infinite-distance divisors lies at infinite distance with respect to the \textit{approximate} metric $G_{i\bar{\jmath}}^{(0)}$.

There are, however, two important difficulties with this approach which we would like to emphasize. 
\begin{itemize}
    \item \textbf{Difficulty (1): Metric-essential instantons}\\
    It can happen that the approximate metric $G_{i\bar{\jmath}}^{(0)}$ derived from the polynomial contribution $P(s)$ is degenerate, and therefore fails to define a positive-definite metric. For example, for a linear polynomial
    \begin{equation}
        P(s) = \alpha s_1 + \beta s_2\,,\qquad \alpha,\beta>0\,,
    \end{equation}
    the induced approximate metric in the saxionic directions is 
    \begin{equation}
        G_{i\bar{\jmath}}^{(0)} = \frac{1}{4(\alpha s^1+\beta s^2)^2} \begin{pmatrix}
            \alpha^2 & \alpha \beta\\
            \alpha \beta & \beta^2
        \end{pmatrix}\,,
    \end{equation}
    which has a vanishing determinant $\det G^{(0)}=0$, and hence is not positive definite. In order to obtain a well-defined metric, one must include exponentially suppressed contributions from $E(a,s)$. More generally, essential instanton terms are required away from the large complex structure regime \cite{Bastian:2021eom}. The subset needed to restore the non-degeneracy of the metric is referred to as \textit{metric-essential} \cite{Bastian:2021hpc}. \\
    \item \textbf{Difficulty (2): Controlling exponential corrections}\\
    Even if every trajectory approaching a boundary point has infinite length with respect to a leading-order metric, exponentially suppressed corrections may change this conclusion. This is particularly relevant when the leading-order metric has null directions as in the situation described above. To illustrate this, consider the following modification of Lee's example \cite[Example~3.10]{Lee2018}:
\begin{equation}\label{eq:metric-Lee}
    ds^2
    =
    \frac{
        (\mathrm da^1-e^{-s^2}\mathrm ds^2)^2
        +
        (\mathrm ds^1+e^{-s^2}\mathrm da^2)^2
        +
        e^{-4s^2}\bigl[(\mathrm da^2)^2+(\mathrm ds^2)^2\bigr]
    }{(s^1)^2}\,.
\end{equation}
We emphasize that this is an illustrative Riemannian metric, it is not a Weil--Petersson metric constructed from a Calabi--Yau family. Furthermore, we consider the trajectory
\begin{equation}\label{eq:trajectory-Lee}
    a^1=\text{constant}\,,
    \qquad
    s^1=s^2=\lambda\,,
    \qquad
    a^2=-e^\lambda\,,
    \qquad
    \lambda\to\infty\,.
\end{equation}
The important observation is that, due to the exponentially large winding in the axionic $a^2$ direction, we have that 
\begin{equation}
    \mathrm ds^1+e^{-s^2}\mathrm da^2
    =
    \mathrm d\lambda-e^{-\lambda}e^\lambda\mathrm d\lambda
    =
    0\,.
\end{equation}
The remaining terms give, for $\lambda_0>0$,
\begin{equation}
    L(\gamma)
    =
    \int_{\lambda_0}^{\infty}
    \frac{e^{-\lambda}\sqrt{2+e^{-2\lambda}}}{\lambda}\,
    \mathrm d\lambda
    <\infty\,.
\end{equation}
In contrast, dropping the exponentially suppressed terms in $s^2$ from the metric \eqref{eq:metric-Lee} gives
\begin{equation}
    ds_{(0)}^2
    =
    \frac{(\mathrm da^1)^2+(\mathrm ds^1)^2}{(s^1)^2}\,.
\end{equation}
This leading-order metric is positive semidefinite, with null directions along $a^2$ and $s^2$. In particular, by the same computation as in Section~\ref{subsubsec:codim-1}, the length $L_0(\gamma)$ of any trajectory $\gamma$ approaching $s^1,s^2\to\infty$ is infinite with respect to this leading metric. The crucial point is that exponential suppression of the corrections to the metric coefficients does not imply suppression of their contribution along an arbitrary trajectory. 
\end{itemize}

\subsection{A supergravity criterion for infinite distance}
\label{subsec:sugra-criterion}
The issues raised in the previous section motivate a different approach to establish whether a given point lies at infinite distance. Rather than classifying all components of the asymptotic metric and attempting to have control over the corrections, our strategy is to instead search for a bound on the metric in terms of a scalar quantity, which may be much easier to control.  Interestingly, such a bound is precisely provided by the BPS mass $\mathcal{M}_q$, as is captured by the following Proposition. 

\begin{subbox}{\begin{prop}\label{prop:distance_criterion}
    A sufficient criterion for infinite distance
\end{prop}}
    Let $q\in H^3(Y,\mathbb{C})$ be an auxiliary charge vector such that, near a boundary point $p_\infty\in\partial\mathcal{M}$, the following conditions hold:
    \begin{align}
        \label{eq:charge-to-mass_condition}&\text{(1)}:\qquad \frac{\mathcal{Q}_q}{\mathcal{M}_q} \leq C\quad \text{throughout a neighbourhood of $p_\infty$}\,,\\
        \label{eq:mass_condition}&\text{(2)}:\qquad \mathcal{M}_q\to 0 \quad \text{as $p_\infty$ is approached}\,,
    \end{align}
    where $C>0$ is constant. Then $p_\infty$ lies at infinite distance with respect to the Weil--Petersson metric \eqref{eq:WP-metric}. More precisely, for any curve $\gamma$ connecting two points $p_1$ and $p_2$, its Weil--Petersson length $L(\gamma)$ satisfies
    \begin{equation}\label{eq:WP_bound}
        L(\gamma)\geq \frac{1}{C}\left|\log \frac{\mathcal{M}_q(p_2)}{\mathcal{M}_q(p_1)} \right|\,.
    \end{equation}
\end{subbox}
More generally, Proposition \ref{prop:distance_criterion} applies to effective variations of Hodge structure of weight $D$ with $h^{D,0}=1$. In this setting, $\Omega$ denotes a holomorphic generator of the top Hodge filtration $F^D$, and the functions $\mathcal{M}_q$ and $\mathcal{Q}_q$ are defined in terms of the polarization and Hodge norm as in \eqref{eq:mass-charge-Hodge} below. As we will see, the proof only uses these Hodge-theoretic quantities and therefore does not require the variation to arise from a family of Calabi--Yau threefolds. This is the formulation that will be used in Section \ref{sec:Hodge}. 

The point of Proposition \ref{prop:distance_criterion} is that it reduces the question of studying infinite distance limits to the question of whether an auxiliary charge vector with the properties \eqref{eq:charge-to-mass_condition} and \eqref{eq:mass_condition} exists. In Section \ref{sec:Hodge}, we will employ methods from degenerations of Hodge structure to establish a precise criterion (Theorem \ref{thm:infinite-distance}) for when such a $q$ indeed exists.   

\subsubsection*{Remarks}
Before describing the proof of Proposition \ref{prop:distance_criterion}, let us make a few additional remarks. 
\begin{itemize}
    \item The criterion only requires the existence of a charge vector $q$ that satisfies the stated conditions. It requires neither the existence nor the stability of a corresponding BPS state. In particular, $q$ need not be integral or real-valued. For this reason, we refer to $q$ as an ``auxiliary'' charge vector. Nevertheless, we will informally refer to the functions $\mathcal{M}_q$ and $\mathcal{Q}_q$ as the mass and charge, and will sometimes speak of the ``would-be'' state associated with $q$. 
    \item For a physical BPS state, the boundedness of the charge-to-mass ratio \eqref{eq:charge-to-mass_condition} implies that its gauge interaction strength does not become parametrically larger than its gravitational interaction strength. The state therefore remains coupled to the gravitational sector in the asymptotic limit, rather than becoming part of a rigid field-theory sector that decouples from gravity \cite{Marchesano:2023thx,Marchesano:2024tod,Castellano:2024gwi,Castellano:2026bnx,Aoufia:2026mqb}.
    \item A useful consequence of Proposition \ref{prop:distance_criterion} is that, if for some  $\mathcal{M}_q\to 0$ along a \textit{finite} length trajectory for some charge vector $q$, then its charge-to-mass ratio cannot remain bounded along that trajectory. A familiar example is provided by the BPS hypermultiplet arising from a D3-brane wrapping the vanishing three-cycle at a conifold point \cite{Strominger:1995cz}, whose charge-to-mass ratio indeed diverges in the limit. 
    \item The extremality bound 
    \begin{equation}
        \mathcal{M}_q \leq \sqrt{2}\mathcal{Q}_q
    \end{equation}
    implies that one may replace condition (2) by the condition that $\mathcal{Q}_q\to 0$ in the limit. This is in fact what we will use in Section \ref{sec:Hodge}. 
\end{itemize}

\subsubsection*{Proof of Proposition \ref{prop:distance_criterion}}
As a first step, we note that the BPS mass $\mathcal{M}_q$, the physical charge $\mathcal{Q}_q$, and the Weil--Petersson metric $G_{i\bar{\jmath}}$ can be written as 
\begin{equation}\label{eq:mass-charge-Hodge}
    \mathcal{M}_q = \frac{|(q,\Omega)|}{\|\Omega\|}\,,\qquad \mathcal{Q}_q = \frac{1}{\sqrt{2}}\|q\|\,,\quad G_{i\bar{\jmath}} = \frac{\langle D_i\Omega,D_j\Omega\rangle}{\|\Omega\|^2}\,,
\end{equation}
where $D_i\Omega = \partial_i\Omega +(\partial_i K)\Omega$ denotes the K\"ahler-covariant derivative, and we have introduced the notation
\begin{equation}
    (v,w) = \int_Y v\wedge w\,,\quad \langle v,w\rangle = \int_Y v\wedge\star\,\bar{w}\,,\qquad \|v\|^2 = \langle v, v\rangle\,,
\end{equation}
for $v,w\in H^3(Y,\mathbb{C})$, which respectively denote the polarization form, the Hodge inner product, and the Hodge norm on $H^3(Y,\mathbb{C})$. 

To proceed, we introduce the real-valued function
\begin{equation}\label{eq:def:H}
    H=\log\mathcal{M}_q\,.
\end{equation}
The function $H$ will play the role of a Lipschitz height function for the metric, as we will now explain. Since
\begin{equation}
    H=\frac12 K+\frac12\log |(q,\Omega)|^2
\end{equation}
and $(q,\Omega)$ is holomorphic, we find
\begin{align}
    \partial_i H
    &=
    \frac12\partial_i K
    +\frac12\frac{(q,\partial_i\Omega)}{(q,\Omega)}
    \notag\\
    \label{eq:partialH}&=
    \frac12\frac{(q,D_i\Omega)}{(q,\Omega)}.
\end{align}
A similar result for the anti-holomorphic derivative $\partial_{\bar{\imath}}H$ follows by complex conjugation. Using \eqref{eq:partialH}, we can then evaluate the derivative of $H$ along any trajectory $\lambda\mapsto z^i(\lambda)$ as follows
\begin{align}
\left|\frac{\mathrm{d}H}{\mathrm{d}\lambda}\right| &= \left|(\partial_i H )\frac{\mathrm{d}z^i}{\mathrm{d}\lambda}+(\partial_{\bar{\imath}}H)\frac{\mathrm{d}\bar{z}^{\bar{\imath}}}{\mathrm{d}\lambda}\right|\,,\\
&\leq 2\left|(\partial_i H)\frac{\mathrm{d}z^i}{\mathrm{d}\lambda}\right|\,,\\
&= \left|\frac{(q,D_i\Omega)}{(q,\Omega)} \frac{\mathrm{d}z^i}{\mathrm{d}\lambda}\right|\,.
\end{align}
For notational simplicity, we introduce 
\begin{equation}
    V = \frac{\mathrm{d}z^i}{\mathrm{d}\lambda} D_i\Omega\,,
\end{equation}
whose Hodge norm satisfies
\begin{equation}
    \|V\|^2 = \left\langle D_i\Omega,D_j\Omega\right\rangle \frac{\mathrm{d}z^i}{\mathrm{d}\lambda} \frac{\mathrm{d}\bar{z}^{\bar{\jmath}}}{\mathrm{d}\lambda} = \|\Omega\|^2 G_{i\bar{\jmath}}\frac{\mathrm{d}z^i}{\mathrm{d}\lambda} \frac{\mathrm{d}\bar{z}^{\bar{\jmath}}}{\mathrm{d}\lambda}\,,
\end{equation}
where we have used \eqref{eq:mass-charge-Hodge} to identify the Weil--Petersson metric. Finally, using the Cauchy--Schwarz inequality, we may write\footnote{This includes the step $(q,V)=-\langle q,\star \bar{V}\rangle$ (recall $\star^2=-1$), and using the fact that the Hodge star and complex conjugation preserve the Hodge norm.}
\begin{align}
    \left|\frac{\mathrm{d}H}{\mathrm{d}\lambda}\right|&\leq \frac{\|q\| \|V\|}{|(q,\Omega)|}\,,\\
    &= \frac{\|q\|\|\Omega\|}{|(q,\Omega)|} \sqrt{G_{i\bar{\jmath}}\frac{\mathrm{d}z^i}{\mathrm{d}\lambda} \frac{\mathrm{d}\bar{z}^{\bar{\jmath}}}{\mathrm{d}\lambda}}\,.
\end{align}
Recalling \eqref{eq:mass-charge-Hodge} and \eqref{eq:def:H}, we thus find
\begin{equation}
\label{eq:distance_bound}
    \boxed{
    \left|\frac{\mathrm{d}\log \mathcal{M}_q}{\mathrm{d}\lambda}\right|\leq  \left(\frac{\mathcal{Q}_q}{\mathcal{M}_q}\right)\sqrt{2G_{i\bar{\jmath}} \frac{\mathrm{d}z^i}{\mathrm{d}\lambda}\frac{\mathrm{d}\bar{z}^{\bar{\jmath}}}{\mathrm{d}\lambda} }}
\end{equation}
A similar relation has previously been studied in  \cite{Gendler:2020dfp,Bastian:2020egp}, for trajectories with unit velocity, in order to obtain a relation between charge-to-mass ratios and the exponential decay rate of BPS particles.\footnote{More precisely, the analysis of \cite{Gendler:2020dfp,Bastian:2020egp} employs the supergravity relation
\begin{equation}\label{eq:sugra-condition}
    \mathcal{Q}_q^2=|\mathcal{Z}_q|^2+4G^{i\bar{\jmath}}\partial_i\mathcal{Z}_q\overline{\partial_j\mathcal{Z}_q}\,,
\end{equation}
which yields a slightly \textit{stronger} bound (here we again keep the trajectory non-normalized)
\begin{equation}\label{eq:bound-adjusted}
    \left|
        \frac{\mathrm{d}\log\mathcal{M}_q}{\mathrm{d}\lambda}
    \right|
    \leq
    \sqrt{
        \left(\frac{\mathcal{Q}_q}{\mathcal{M}_q}\right)^2-1}
    \sqrt{
        G_{i\bar{\jmath}}
        \frac{\mathrm{d}z^i}{\mathrm{d}\lambda}
        \frac{\mathrm{d}\bar z^{\bar{\jmath}}}{\mathrm{d}\lambda}
    }.
\end{equation}
However, the relation \eqref{eq:sugra-condition} is specific to variations of Hodge structure of weight $D=3$. Instead, our analysis does not impose a restriction on the weight. 
} 
For our purposes, it is useful to keep the trajectory non-normalized in order to use the relation \eqref{eq:distance_bound} as a lower bound on the distance. Indeed, recalling the assumed boundedness of the charge-to-mass ratio \eqref{eq:charge-to-mass_condition} and integrating \eqref{eq:distance_bound}, we obtain
\begin{align}
    L(\gamma) &= \int_{\lambda_1}^{\lambda_2}\sqrt{2G_{i\bar{\jmath}} \frac{\mathrm{d}z^i}{\mathrm{d}\lambda}\frac{\mathrm{d}\bar{z}^{\bar{\jmath}}}{\mathrm{d}\lambda} } d\lambda\,,\\
    &\geq \frac{1}{C}\int_{\lambda_1}^{\lambda_2} \left|\frac{\mathrm{d}\log\mathcal{M}_q}{\mathrm{d}\lambda}\right| d\lambda\,,\\
    &\geq \frac{1}{C}\left|\int_{\lambda_1}^{\lambda_2} \frac{\mathrm{d}\log\mathcal{M}_q}{\mathrm{d}\lambda}d\lambda\right|\,,\\
    &= \frac{1}{C} \left|\log\frac{\mathcal{M}_q(p_2)}{\mathcal{M}_q(p_1)}\right|\,,
\end{align}
as we wanted to show. In particular, if the trajectory approaches the boundary point $p_\infty$, the assumption \eqref{eq:mass_condition} that $\mathcal{M}_q\to 0$ implies that $L(\gamma)\to\infty$. Hence, the boundary point lies at infinite distance with respect to the Weil--Petersson metric. This completes the proof of Proposition \ref{prop:distance_criterion}.

\subsection{Comments on the Swampland Distance Conjecture}
\label{subsec:SDC}
The distinction between finite and infinite distance limits is of central importance for the Swampland Distance Conjecture. For reference, let us recall the precise statement.

\begin{subbox}{Swampland Distance Conjecture \cite{Ooguri:2006in}}
Consider a theory that is consistent with quantum gravity, with a moduli space $\mathcal{M}$ parametrized by the vacuum expectation values of the massless scalar fields in the theory. Then the following statements hold. 
\begin{enumerate}
    \item For any point $p_0\in\mathcal{M}$, there exists another point $p_\infty\in\partial\mathcal{M}$ such that the (geodesic) distance between $p_0$ and $p_\infty$ is infinite. 
    \item For any trajectory $\gamma$ connecting two points $p_0$ and $p_\infty$, there exists an infinite tower of states whose associated mass scale $M$ becomes exponentially small in the geodesic distance:
    \begin{equation}
        \frac{M(p)}{M(p_0)}\sim e^{- \alpha L(\gamma)}\,,\qquad \text{as $p\to p_\infty$}\,.
    \end{equation}
\end{enumerate}
\end{subbox}
In the following, we make two small comments pertaining to the first and second statements in the Swampland Distance Conjecture, respectively. While these observations are known in the literature, we feel that they are worth emphasizing in the present context.

\subsubsection{Absence of infinite-distance limits}
The main goal of this work is to address \textit{when} a given boundary point $p_\infty\in\partial\mathcal{M}$ lies at infinite distance. It does not, however, address the question whether an infinite distance limit always \textit{exists}. Indeed, it should be stressed that there exist concrete examples of Calabi--Yau manifolds whose complex structure moduli space does not have any infinite-distance boundary. A trivial such example is given by any rigid Calabi--Yau manifold, which has $h^{2,1}=0$ and thus has a zero-dimensional complex structure moduli space. A more non-trivial example is given by the ``orphan'' double-octic arrangement No.~153 of \cite{CynkVanStraten2019}, see also \cite{CynkVanStraten2022} for a more detailed analysis of this particular geometry and \cite{Joshi:2019nzi,Cecotti:2020rjq} for related discussions. The example has $h^{2,1}=1$, and its parameter space is a four-punctured Riemann sphere $\mathbb{P}^1\setminus\{0,-1,-2,\infty\}$. The Riemann symbol of the associated Picard--Fuchs system is given by
\begin{equation}
P\left\{
\begin{array}{cccc}
0 & -1 & -2 & \infty \\[1.0mm]\hline 
0 & 0 & 0 & \frac12\\[1mm]
\frac12 & \frac12 & \frac12 & 1\\[1mm]
\frac12 & \frac12 & \frac32 & 1\\[1mm]
1 & 1 & 2 & \frac32
\end{array}
\right\}.
\end{equation}
Importantly, all four singularities are at finite distance with respect to the Weil--Petersson metric. Indeed, $\{0,-1,\infty\}$ are conifold-type degenerations, while the remaining singularity at $-2$ has a finite-order monodromy and thus becomes trivial after a base change.

Note that this is not in tension with the Swampland Distance Conjecture as stated above, since the full moduli space also contains the hypermultiplet scalars. In particular, the weak-coupling limit $g_s\to 0$ (at fixed and sufficiently large K\"ahler moduli) always exists as an infinite-distance limit in this setting, even when taking into account quantum corrections to the hypermultiplet moduli space \cite{Marchesano:2019ifh,Baume:2019sry}. 

\subsubsection{Exponential rates and geodesics vs.~non-geodesics}
The bound \eqref{eq:WP_bound} can be written as
\begin{equation}
    e^{-C L(\gamma)}\leq \frac{\mathcal{M}_q(p_2)}{\mathcal{M}_q(p_1)} \leq e^{C L(\gamma)}\,.
\end{equation}
In particular, the left-hand side of this relation implies that the mass of any would-be state satisfying \eqref{eq:charge-to-mass_condition} cannot decay \textit{faster} than exponentially with the traversed field distance. Interestingly, a similar conclusion was reached in \cite{vandeHeisteeg:2023ubh} by analysing bounds on the variation of the species scale. Note that this is precisely opposite to the bound predicted by the Swampland Distance Conjecture, which states that the mass cannot decrease \textit{slower} than exponentially in the traversed field distance, along \textit{geodesic} trajectories.  

It should be stressed, however, that the mass can also decay much slower than exponentially along ``Swampy'' trajectories, as discussed in detail in \cite{Calderon-Infante:2020dhm}. In particular, such trajectories are typically non-geodesic, but instead feature large winding in the axionic directions. As an example, consider a metric of the form
\begin{equation}
    ds^2 = \frac{3}{2} \frac{\mathrm{d}a^2+\mathrm{d}s^2}{s^2}\,,
\end{equation}
which arises (asymptotically) in the one-parameter LCS/LV regime. By the same arguments as above, the boundary $s\to\infty$ lies at infinite distance. However, along a trajectory of the form
\begin{equation}
    \gamma:\qquad a(\lambda) = \lambda^2\,,\quad s(\lambda) = \lambda\,,
\end{equation}
the distance scales as $L(\gamma) = \sqrt{6}\lambda+\mathcal{O}(1)$, such that the mass of the D0-brane state with period $(q,\Omega)$ and $\mathcal{Q}_q/\mathcal{M}_q=2$ scales as\footnote{In these conventions, the D0-brane mass scales as $s^{-3/2}$.}
\begin{equation}
    \mathcal{M}_q \sim L(\gamma)^{-3/2}\,,
\end{equation}
which is polynomially small in the distance rather than exponentially small.

\section{A Hodge-theoretic criterion for infinite distance}
\label{sec:Hodge}
In this section, we describe a precise Hodge-theoretic theorem establishing when a boundary point lies at infinite distance. The key step is to construct a single auxiliary charge whose Hodge norm
vanishes and whose charge-to-mass ratio remains bounded in every asymptotic growth sector. In Section~\ref{subsec:hodge-background}, we recall the minimal background on degenerations of Hodge structures that is necessary to state our main result in Theorem \ref{thm:infinite-distance}. Then, in Section~\ref{subsec:proof}, we discuss the asymptotics of Hodge norms using the growth theorem of Cattani, Kaplan, and Schmid, and subsequently prove Theorem~\ref{thm:infinite-distance}. Some additional machinery, including Deligne splittings and the associated $\mathfrak{sl}(2)$-algebras, is used to prove Lemma \ref{lem:weights} and is reviewed in Appendix \ref{app-lemma}. For further details and references on various aspects of asymptotic Hodge theory, we refer the reader to \cite{vandeHeisteeg:2022gsp,Monnee:2024gsq}. 

\subsection{Degenerations of Hodge structures and the main theorem}
\label{subsec:hodge-background}

\subsubsection*{Variations of Hodge structure of Calabi--Yau type}
Let us consider an effective real polarized variation of Hodge structure of weight $D\geq1$. Since we are interested in a local description near a codimension-$n$ intersection of normal-crossing boundary divisors, we work over a parameter space of the form
\begin{equation}
    (\Delta^*)^n\times\Delta^m,
\end{equation}
where $\Delta$ denotes a disk and $\Delta^*=\Delta\setminus\{0\}$. As before, we choose coordinates $z^i$ on the punctured disks and $\zeta^\alpha$ on the remaining disks. Throughout, we will fix a boundary point $p_\infty\in \Delta_I^{\circ}$, with $|I|=n$, and choose the spectator coordinates $\zeta$ such that $p_\infty$ corresponds to $(0,0)$. In fact, as explained in \cite{CattaniKaplan1989}, we may and will assume that $m=0$ without loss of generality, i.e.~we may neglect the spectator coordinates. The reason is that all statements concerning Hodge norms will hold uniformly on compact subsets of $\Delta^m$. 

Effectivity means that the Hodge decomposition takes the form
\begin{equation}
    H_{\mathbb C}=\bigoplus_{r=0}^D H^{r,D-r},
    \qquad \overline{H^{r,D-r}}=H^{D-r,r}.
\end{equation}
We further assume $h^{D,0}=1$, and refer to such a variation as being of \emph{Calabi--Yau $D$-fold type}. Note that this terminology does not require the variation to arise from a family of Calabi--Yau manifolds.

The polarization is a flat non-degenerate $(-1)^D$-symmetric pairing $(\cdot,\cdot)$ satisfying the Hodge--Riemann bilinear relations. It will be convenient to work with the Hodge filtration
\begin{equation}
    F^p=\bigoplus_{r=p}^D H^{r,D-r},
\end{equation}
which varies holomorphically and satisfies Griffiths transversality. Let $\Omega$ be a holomorphic generator of $F^D$, and denote the Hodge norm by $\|\cdot\|$. As before, the Weil--Petersson metric is given locally by
\begin{equation}
    G_{i\bar{\jmath}}=\partial_i\bar{\partial}_{\bar{\jmath}}K,
    \qquad K=-\log\|\Omega\|^2.
\end{equation}
To be precise, for an abstract variation of Hodge structure of the type we consider, this metric is only positive semidefinite \cite{Wang1997}. However, it can still be used to measure the length of trajectories. 

\subsubsection*{Monodromy}
Of particular importance is the behaviour of the Hodge filtration when encircling a boundary divisor. Let $T_i$ denote the monodromy operator associated with a loop $z^i\mapsto z^i e^{2\pi \mathrm{i}}$ around $\Delta_i$. We will assume that the $T_i$ are quasi-unipotent. For \textit{integral} variations of Hodge structure this is guaranteed by the Monodromy Theorem of \cite{Landman1973,schmid}, see also \cite{CattaniKaplan1989}. In particular, this is the case in the geometric setting described in Section~\ref{sec:physical-criterion}. Then, after possibly performing a base change, we may assume the $T_i$ are unipotent, such that their logarithms
\begin{equation}
    N_i=\log T_i,
    \qquad [N_i,N_j]=0,
\end{equation}
are real nilpotent operators that preserve the polarization infinitesimally
\begin{equation}\label{eq:preserve-polarization}
    (N_i v,w)+(v,N_i w)=0\,.
\end{equation}
As before, we employ the coordinates
\begin{equation}
    t^i=\frac{\log z^i}{2\pi\mathrm{i}}=a^i+\mathrm{i}s^i\,.
\end{equation}
In these coordinates, the monodromy is generated by sending $a^i\mapsto a^i+1$. By definition, approaching the boundary point $p_\infty$ implies $s^i\to\infty$ for every $i=1,\ldots, n$.

\subsubsection*{The nilpotent orbit theorem}
The nilpotent orbit theorem of Schmid allows us to write \cite{schmid}
\begin{equation}\label{eq:nilpotent-top-form}
    \Omega=\mathrm{exp}\left[\sum_i t^i N_i \right]A(z)\,,
    \qquad A(0)=a_0\neq0\,,
\end{equation}
where $A$ is holomorphic in $z^i$, in particular it extends holomorphically to the puncture of the disks. 

\subsubsection*{The main theorem}
Before stating the theorem, we introduce some additional notation. For each permutation $\sigma$ of $\{1,\ldots,n\}$, we define
\begin{equation}\label{eq:partial-monodromies-degrees}
    N_k^\sigma=\sum_{j=1}^k N_{\sigma(j)},
    \qquad
    d_k^\sigma=\max\{r\in\mathbb{Z}_{\geq 0}:(N_k^\sigma)^r a_0\neq0\}\,.
\end{equation}
Here, $d_k^\sigma$ counts the largest number of times that $N_k^\sigma$ can act on $a_0$ without annihilating it. Effectivity and the fact that the variation of Hodge structure is polarized together imply that $0\leq d_k^\sigma\leq D$. These integers play a central role in classifying the degeneration associated with sending $s_{\sigma(1)},\ldots, s_{\sigma(k)}\to\infty$. For $D=3$, the values $d_k^\sigma=0,1,2,3$ correspond to Type I, II, III and IV degenerations, respectively \cite{Grimm:2018ohb,Grimm:2018cpv}. We note that, for $k=n$, the sum is independent of the permutation and we employ the simplified notation
\begin{equation}\label{eq:def-N}
    N=N_n^\sigma=\sum_{i=1}^nN_i\,,
    \qquad d=d_n^\sigma\,.
\end{equation}
\newpage

We can now state the result to be proved in the next subsection.
\begin{subbox}{\begin{thm}\label{thm:infinite-distance}{A Hodge-theoretic criterion for infinite distance}\end{thm}}
    Consider an effective polarized real variation of Hodge structure of weight $D$ with $h^{D,0}=1$, with unipotent local monodromies. Suppose that
    \begin{equation}
        Na_0\neq0\,,
    \end{equation}
    or equivalently $d>0$. Then every trajectory $\gamma$ approaching $p_\infty$ has infinite length
    with respect to the Weil--Petersson metric.
\end{subbox}
For $D=3$, Theorem~\ref{thm:infinite-distance} precisely gives the distance statement for Type II, III, and IV limits stated in the Introduction. 

\subsection{Proof of Theorem~\ref{thm:infinite-distance}}
\label{subsec:proof}
Before we give a detailed description of the proof, let us first recall the main strategy based on the insights of Section \ref{sec:physical-criterion}. In particular, according to Proposition \ref{prop:distance_criterion}, it suffices to find an element $q\in H_{\mathbb{C}}$ such that
\begin{align}
    \label{eq:cond0}&\text{(0)}:\quad &&(q,\Omega)\neq 0\,,\\
    \label{eq:cond1}&\text{(1)}:\qquad &&\frac{\|q\|\|\Omega\|}{|(q,\Omega)|}\quad \text{is bounded in a neighbourhood of $p_\infty$}\,,\\
    \label{eq:cond2}&\text{(2)}:\quad &&\|q\|\to 0 \quad \text{as $p_\infty$ is approached.}
\end{align}
For clarity, here we have emphasized the condition (0) that $(q,\Omega)\neq 0$. In addition, as explained in the remarks after Proposition \ref{prop:distance_criterion}, we impose $\|q\|\to 0$ as a sufficient condition for $\mathcal{M}_q\to 0$, by using the Cauchy--Schwarz inequality
\begin{equation}
    \frac{|(q,\Omega)|}{\|\Omega\|}\leq \|q\|\,.
\end{equation}
Clearly, our main task in order to verify the above conditions for some auxiliary charge $q$ is to obtain uniform bounds on the Hodge norm $\|\cdot\|$ in a neighbourhood of $p_\infty$. The necessary information is captured by the growth theorem of Cattani, Kaplan, and Schmid, as we now recall. 

\subsubsection{Growth sectors and growth theorems}
\subsubsection*{Weight filtrations}

The relevant construction is the monodromy weight filtration associated to the nilpotent operators $N_k^\sigma$. In general, given a nilpotent operator $N$ (not necessarily the one in equation \eqref{eq:def-N}), there is a unique increasing
filtration $W_\ell(N)$, for $\ell=0,\ldots, 2D$, such that
\begin{align}
    NW_\ell(N)&\subseteq W_{\ell-2}(N),
    \\
    N^r:\operatorname{Gr}^{W(N)}_{D+r}
    &\xrightarrow{\ \sim\ }
    \operatorname{Gr}^{W(N)}_{D-r},
    \qquad r\geq0,
\end{align}
where $\operatorname{Gr}^W_\ell=W_\ell/W_{\ell-1}$. The first condition says that $N$ lowers the weight by two, while the second gives a linear isomorphism between the graded pieces at equal distance from weight $D$. An important property is that the weight filtration of a positive linear combination of the $N_i$ is independent of its
coefficients. In particular, each ordering determines the filtrations
\begin{equation}
    W^{\sigma,k}_\bullet=W_\bullet(N_k^\sigma),
    \qquad k=1,\ldots,n.
\end{equation}
In general, these are different filtrations. 

\subsubsection*{Growth sectors and Hodge norms}

In a multi-parameter limit, the growth of the Hodge norm depends on the relative sizes of the saxions. It is therefore useful to divide the asymptotic region into sectors with a fixed ordering. For each permutation $\sigma$, we define
\begin{equation}\label{eq:ordered-growth-sector}
    R_\sigma=
    \left\{t:\quad
        s^{\sigma(1)}\geq\cdots\geq s^{\sigma(n)}\geq A,
        \quad 0\leq a^i\leq1
    \right\}\,,
\end{equation}
where $A>1$ is a constant. Taking the union over all growth sectors covers the region in which the saxions satisfy the lower bound $\min_i s^i\geq A$, while the axions are restricted to a fundamental domain. 

The relation between weight filtrations and Hodge norms is made precise by the growth theorem of Cattani, Kaplan, and Schmid \cite{schmid,CKS}, see also \cite{BakkerKlinglerTsimerman2020}. For our purposes, it suffices to state the following consequence. Consider a fixed vector $v\in H_{\mathbb C}$ satisfying 
\begin{equation}
    v\in\bigcap_{k=1}^n W_{\ell_k}^{\sigma,k}\,.
\end{equation}
Then its \textit{exact} Hodge norm is bounded on $R_\sigma$ by
\begin{equation}\label{eq:growth-theorem}
    \|v\|^2\leq C_{\sigma}
    \prod_{k=1}^{n-1}
    \left(\frac{s^{\sigma(k)}}{s^{\sigma(k+1)}}\right)^{\ell_k-D}
    \left(s^{\sigma(n)}\right)^{\ell_n-D}\,,
\end{equation}
where $C_{\sigma}>0$ is a constant that may depend on the choice of growth sector (in particular the constant $A$ in \eqref{eq:ordered-growth-sector}). It can also depend on the choice of $v$, but since this is kept fixed throughout we will omit the dependence in the notation. 

\subsubsection{Proof of Theorem~\ref{thm:infinite-distance}}
We now proceed with the proof of Theorem \ref{thm:infinite-distance}. To start, consider the auxiliary charge vector
\begin{equation}
    q = N^d \overline{a_0} \in H_\mathbb{C}\,,
\end{equation}
where we recall that $N$ is defined in \eqref{eq:def-N}, and $a_0$ is defined in \eqref{eq:nilpotent-top-form} and arises in the nilpotent orbit expansion of $\Omega$. We will show that $q$ satisfies the desired conditions (0), (1), and (2) listed in \eqref{eq:cond0}--\eqref{eq:cond2}. 

Let us start with the observation that $(q,a_0)$ is non-zero. This follows directly from the properties of polarized mixed Hodge structures. Indeed, we have that $a_0\in\mathrm{Gr}^{W(N)}_{D+d}$ and is primitive, such that the Hodge--Riemann bilinear relations for the primitive polarization $(N^d\cdot,\cdot)$ imply that
\begin{equation}\label{eq:q-pairing}
    (q,a_0) = (N^d\overline{a}_0, a_0)\neq 0\,.
\end{equation}
The remainder of the proof requires obtaining estimates of the Hodge norms $\|q\|$ and $\|\Omega\|$, which are provided by the growth theorem. Thus, it is in particular necessary to know the location of $a_0$ and $q$ within all of the weight filtrations associated with the various growth sectors. To this end, we state the following technical Lemma, whose proof we delegate to Appendix~\ref{app-lemma} for presentational clarity. 

\begin{lem}\label{lem:weights}
    For every permutation $\sigma\in S_n$ and every $1\leq k\leq n$, we have
\begin{equation}
    a_0\in W(N_k^\sigma)_{D+d_k^\sigma},
    \qquad
    q\in W(N_k^\sigma)_{D-d_k^\sigma}.
\end{equation}
Furthermore,
\begin{equation}
    N_iq=0
    \qquad\text{for all }i=1,\ldots,n.
\end{equation}
\end{lem}
In particular, we emphasize that $q$ is monodromy-invariant, and that its weights are precisely opposite to those of $a_0$. 

\subsubsection*{Condition (2): Vanishing $\|q\|$}
By direct application of the growth estimate~\eqref{eq:growth-theorem} to $q$, together with Lemma~\ref{lem:weights}, we find
\begin{equation}\label{eq:q-growth}
    \|q\|^2 \leq C_\sigma H_\sigma\,,\qquad \text{on the growth sector $R_\sigma$}\,,
\end{equation}
where $C_\sigma>0$ is a constant, and we have introduced the function
\begin{equation}\label{eq:def-Hsigma}
    H_\sigma = \left[\prod_{k=1}^{n-1} \left( \frac{s^{\sigma(k)}}{s^{\sigma(k+1)}}\right)^{-d_k^\sigma}\right] \left(s^{\sigma(n)}\right)^{-d}\,.
\end{equation}
Note that each ratio appearing in \eqref{eq:def-Hsigma} is at least 1, and each exponent $-d^\sigma_k$ is nonpositive. Therefore, we can discard these factors to obtain a bound
\begin{equation}\label{eq:q-growth_bound}
    \|q\|^2 \leq C_\sigma\left(s^{\sigma(n)}\right)^{-d}\,,\qquad \text{on the growth sector $R_\sigma$}\,.
\end{equation}
We stress that the above bound holds within the growth sector $R_\sigma$. We can, however, turn \eqref{eq:q-growth_bound} into a uniform bound that is valid across all growth sectors at once as follows. First, by definition of the growth sectors, we have the ordering $s^{\sigma(1)}\geq \cdots \geq s^{\sigma(n)}$, and hence
\begin{equation}
    s^{\sigma(n)} = \min_i\,s^i\,.
\end{equation}
Second, since there are only finitely many growth sectors, we can simply take the maximum $C=\max_\sigma C_\sigma$ over all permutations to obtain the bound
\begin{equation}\label{eq:q-bound-uniform}
    \boxed{\text{(2)}:\quad \|q\|^2 \leq C \left(\mathrm{min}_i\,s^i \right)^{-d}}
\end{equation}
which is valid uniformly on the union of the growth sectors. Since $d>0$ by assumption, we see that necessarily $\|q\|\to 0$ along any trajectory that approaches the boundary point $p_\infty$.  

To complete the proof of Theorem \ref{thm:infinite-distance}, it remains to show that the ratio
\begin{equation}\label{eq:ratio}
    \frac{\|q\| \|\Omega\|}{|(q,\Omega)|}
\end{equation}
is uniformly bounded. First, we show that the denominator in \eqref{eq:ratio} is bounded away from zero, which will establish condition (0). 

\subsubsection*{Condition (0): Nonvanishing $(q,\Omega)$}
Since $q$ is monodromy-invariant by Lemma \ref{lem:weights} and the monodromy operators preserve the polarization \eqref{eq:preserve-polarization}, it follows that
\begin{align}
    (q,\Omega) &= \left(q ,\mathrm{exp}\left(\sum_{i=1}^n t^iN_i \right) A(z)\right)\,,\\    
    &=\left(\mathrm{exp}\left(-\sum_{i=1}^n t^iN_i \right) q ,A(z)\right)\,,\\
    \label{eq:invariance}&=(q,A(z))\,,\\
    &=(q,a_0)+\cdots\,.
\end{align}
Using the fact that $(q,a_0)\neq 0$, together with the holomorphicity of $A(z)$, we find that $|(q,\Omega)|$ can be uniformly bounded in an open polydisc around $z^i=0$ as
\begin{equation}\label{eq:(q,Omega)-bound}
    \boxed{\text{(0)}:\quad c_- < |(q,\Omega)|<c_+}
\end{equation}
for some constants $c_{\pm}>0$. 

\subsubsection*{Condition (1): Bounded charge-to-mass ratio}
Next, we show that the product $\|q\| \|\Omega\|$ is uniformly bounded. To see this, we use the fact that
\begin{equation}\label{eq:Omega-scaling}
    \|\Omega\|^2 \leq C_\sigma' H_\sigma^{-1}\,,\qquad \text{on the growth sector $R_\sigma$}\,,
\end{equation}
for some constant $C_\sigma'>0$, where we recall the definition of the function $H_\sigma$ in \eqref{eq:def-Hsigma}. This bound deserves some extra explanation, because $\Omega$ in principle carries additional moduli dependence, in contrast to the fixed charge $q$, hence one cannot directly apply the standard growth theorem. However, as argued in the original work \cite{CKS}, see also \cite[Theorem~4.8 and Lemma~4.10]{BakkerKlinglerTsimerman2020}, one can apply the growth theorem to $a_0$ and argue that neither the additional factor $\mathrm{exp}\left( \sum_i t^i N_i\right)$ nor the holomorphic corrections in $A(z)$ affect the asymptotic growth of the Hodge norm. The bound \eqref{eq:Omega-scaling} then follows again from applying Lemma~\ref{lem:weights}. 

Importantly, due to the fact that $\|q\|$ and $\|\Omega\|$ scale precisely oppositely, we find that
\begin{equation}
    \|q\| \|\Omega\|\leq \sqrt{C_\sigma C_\sigma'}\,,\qquad \text{on the growth sector $R_\sigma$}\,.
\end{equation}
Now we apply the same argument as before, taking the maximum $C'=\max_\sigma C'_\sigma$, in order to obtain a uniform bound
\begin{equation}\label{eq:|q||Omega|-bound}
    \|q\| \|\Omega\| \leq \sqrt{C C'}\,.
\end{equation}
Together, equations \eqref{eq:(q,Omega)-bound} and \eqref{eq:|q||Omega|-bound} imply 
\begin{equation}\label{eq:bound-final}
    \boxed{\text{(1)}:\quad \frac{\|q\| \|\Omega\|}{|(q,\Omega)|}<C''\,,\qquad C''=\frac{\sqrt{C C'}}{c_-}}
\end{equation}
Thus, $q$ satisfies the desired conditions \eqref{eq:cond0}--\eqref{eq:cond2}. Therefore, by Proposition~\ref{prop:distance_criterion}, the boundary point $p_\infty$ lies at infinite distance along any trajectory in the union of the growth sectors, i.e.~in
\begin{equation}
    \{\min_i s^i\geq A\,,\quad0\leq a^i\leq 1\}\,,
\end{equation}
for some sufficiently large $A>1$. The situation is  depicted in Figure~\ref{fig:uniform}.

\begin{figure}
    \centering
    \includegraphics[width=0.5\linewidth]{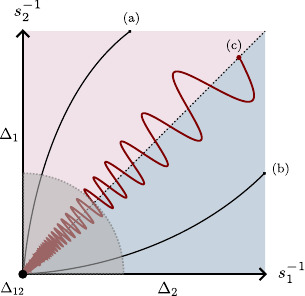}
    \caption{A local depiction near the intersection $\Delta_{12}$ of two boundary divisors $\Delta_1$ and $\Delta_2$, restricted to the saxionic directions. The trajectories (a) and (b) stay within the growth sectors $s^1\geq s^2$ and $s^2\geq s^1$, respectively, while the trajectory (c) oscillates between them. The proof of Theorem~\ref{thm:infinite-distance} constructs an auxiliary charge $q$ for which the boundedness results \eqref{eq:q-bound-uniform} and \eqref{eq:bound-final} hold uniformly in a neighbourhood (depicted in grey) of the boundary point $p_\infty\in\Delta_{12}$. Since trajectory (c) eventually enters this neighbourhood, the proof applies equally well to it.}
    \label{fig:uniform}
\end{figure}

\subsubsection*{Axion winding}
To fully close the proof, we must also address trajectories that feature winding in the axionic directions, i.e.~for which the axions $a^i$ exit the fundamental domain $0\leq a^i\leq 1$. This can be addressed straightforwardly by recalling that the action $a^i\mapsto a^i+1$ precisely corresponds to the monodromy action. Crucially, since $q$ is monodromy-invariant by Lemma~\ref{lem:weights}, and the monodromy action preserves the polarization, the quantities $\|q\|$, $\|\Omega\|$ and $(q,\Omega)$ are all invariant under integer axion shifts.\footnote{To be explicit, under $a^i\mapsto a^i+1$, we have for example that $\|q\|\mapsto \|T_i^{-1} q\|=\|q\|$, where the second step follows from the monodromy-invariance of $q$. Instead, for $\|\Omega\|$ one uses the fact that $\|\cdot\|$ and $\Omega$ transform in precisely the opposite way under monodromy. Finally, for $(q,\Omega)$ it follows from \eqref{eq:invariance}.} Thus, all estimates established above extend from the fundamental domain to all axion values.

Since every curve that approaches $p_\infty$ eventually enters a neighbourhood in which the above estimates hold, this completes the proof of Theorem~\ref{thm:infinite-distance}. Together with the result of Lee \cite[Theorem~2.2]{Lee2018}, this proves Theorem~\ref{thm:main}, thus establishing a necessary and sufficient criterion for infinite distance in full generality. 

\section{Conclusions}
\label{sec:discussion}

In this work, we investigated infinite-distance limits in the complex structure moduli space of Calabi--Yau manifolds. In particular, we established a criterion for infinite distance at intersections of boundary divisors that applies to the full Weil--Petersson metric. The proof is based on the inequality \eqref{eq:distance_bound}, which allowed us to bound the length of a trajectory in terms of the BPS mass $\mathcal{M}_q$. Of crucial importance is the construction of a monodromy-invariant auxiliary charge $q=N^d\overline{a_0}$ for which the charge-to-mass ratio is uniformly bounded in the limit, using the tools of asymptotic Hodge theory. The main technical step is Lemma~\ref{lem:weights}, which determines the weights of $q$ with respect to the weight filtrations for every ordering of the saxions. This allows the estimates in the individual growth sectors to be combined into a bound throughout a neighbourhood of the boundary point. 

Although initially formulated in the language of Type IIB compactifications, the argument applies more generally to effective polarized variations of Hodge structure of Calabi--Yau $D$-fold type under the assumptions of Theorem~\ref{thm:infinite-distance}. From a bottom-up perspective, it is well-known that the special K\"ahler geometry of the vector multiplet sector of four-dimensional $\mathcal{N}=2$ supergravity contains precisely the data of such an effective polarized real variation of Hodge structure of Calabi--Yau threefold type, see for example \cite{Cecotti:2020rjq}. It is therefore an interesting question to ask whether our results apply more generally to that setting. Of vital importance in this regard is to what extent the quasi-unipotency of the monodromy can be relaxed. In fact, Sabbah and Schnell recently established versions of the nilpotent orbit theorem and a simplified $\mathrm{Sl}(2)$-orbit theorem without this assumption for one-parameter degenerations \cite{SabbahSchnell2022}, see also \cite{Deng2022} for related results, as well as \cite{Kaufmann_to_appear} for further discussions concerning bottom-up arguments for the Emergent String Conjecture in this setting. 

Returning to the geometric setting, we note that the rationality of the monodromy operators and their weight filtrations allows one to replace the auxiliary charge $q$ by an \textit{integral} charge with the same properties. Thus, such a charge can, in principle, be realized by a physical state, although integrality alone does not establish the existence of that state.\footnote{For a different use of complex-valued charges in the interpretation of threshold corrections and the Emergence Proposal, see \cite{Hattab:2025aok,Hattab:2026lho}.} In fact, it precisely corresponds to the Kaluza--Klein modes that were identified in \cite{Hassfeld:2025uoy,Monnee:2025ynn} based on a geometric analysis of Calabi--Yau degenerations, in accordance with expectations from the Emergent String Conjecture. Our analysis therefore provides evidence that in particular the enhancement analysis of \cite{Monnee:2025ynn} generalizes to arbitrary trajectories. It is also noteworthy that the distance criterion itself does not require a tower: a \textit{single} state whose mass vanishes and whose charge-to-mass ratio remains uniformly bounded throughout a neighbourhood of the boundary already suffices to show that the limit is at infinite distance.\footnote{Note, however, that we are \textit{not} claiming that integrating out this single state generates the infinite-distance limit.}

\subsubsection*{Bounding the metric through scalar quantities}
A broader lesson of the proof is that exact statements about distances can sometimes be obtained by controlling quantities that are simpler than the metric itself. To make this precise, consider a moduli space with real coordinates $\phi^A$ and a metric $G_{AB}$. Suppose that there exists a function $f$ satisfying
\begin{equation}
    \left|v^A\partial_A\log f\right|\leq C\sqrt{G_{AB}v^Av^B}
    \label{eq:discussion-scalar-bound}
\end{equation}
for every tangent vector $v$, with a constant $C>0$ that is uniform in a neighbourhood of a boundary point. Integrating along a trajectory contained in this neighbourhood then gives
\begin{equation}
    L(\gamma)\geq\frac{1}{C}
    \left|\log\frac{f(p_2)}{f(p_1)}\right|\,.
\end{equation}
In particular, if $f\to0$ or $f\to\infty$ at the boundary point, any trajectory approaching it will have infinite length. To apply this strategy in general, the main task is thus to identify a suitable function $f$ and establish the inequality uniformly. In our setting, the inequality follows from the Cauchy--Schwarz inequality applied to the Hodge inner product, while the required asymptotic input consists of bounds on periods and Hodge norms. Similar reasoning may be useful in other settings where masses, tensions, or other physical observables are more tractable than the full metric, see also \cite{Etheredge:2023zjk} for related ideas. 

\subsubsection*{4d $\mathcal{N}=1$ K\"ahler geometry}
In this regard, it would be interesting to establish analogous statements in four-dimensional compactifications with $\mathcal{N}=1$ supersymmetry. The field space metric for the scalars in the chiral multiplets is K\"ahler, but need not possess the special K\"ahler structure we used in the $\mathcal{N}=2$ vector multiplet sector. Our example in Section~\ref{subsec:limits} is therefore particularly relevant: the K\"ahler property in principle allows two infinite-distance divisors to intersect at finite distance. However, it is not clear whether this behaviour can actually occur in a moduli space arising from a consistent string compactification. In addition, the importance of addressing this question beyond the classical approximation is illustrated by the quantum obstructions recently studied in \cite{Kaufmann:2026fli,Kaufmann:2026mha,Kaufmann:2026tsy}. In these works, classical infinite-distance limits in the complex-structure sector of $\mathcal{N}=1$ compactifications can encounter unsuppressed quantum corrections, see also \cite{Klaewer:2020lfg}. This raises the question of how the corresponding asymptotic regions are described by the full physical metric. An extension of our strategy could provide a way to address this question through suitably controlled masses or tensions, without requiring an explicit determination of the full metric. In this regard, EFT strings \cite{Lanza:2020qmt,Lanza:2021udy} provide a natural starting point to address this question, with their tensions as candidate functions. We hope to return to this question in future work. 

\subsubsection*{Asymptotic geodesics and corrections}
Returning to the Calabi--Yau threefold setting, another natural question is whether one can determine the leading asymptotic behaviour of \textit{geodesics} while retaining control over corrections to the leading approximation to the K\"ahler potential, see also \cite{Blumenhagen:2018nts,Calderon-Infante:2022nxb,Etheredge:2023zjk,Raman:2024fcv,Ruehle:2024ufw} for related work on geodesics. In particular, in recent work \cite{Baines:2025upi} a group-theoretic characterization of geodesics was given in locally symmetric moduli spaces. It would be interesting if this analysis could be extended to the present setting. It should be stressed, however, that in the present setting this question is expected to be more difficult due to the generic presence of metric-essential instantons \cite{Bastian:2021eom,Bastian:2021hpc}. Importantly, since the geodesic equations involve the inverse of the metric and its derivatives, such corrections to the K\"ahler potential may in fact change the leading behaviour of geodesics whenever the leading metric contains null directions. We hope to investigate this question in future work.

\subsection*{Acknowledgements}
We thank Florent Baume, Thomas Grimm, Lukas Kaufmann, and Timo Weigand for interesting discussions and feedback. This work is supported in part by Deutsche Forschungsgemeinschaft under Germany’s Excellence Strategy EXC 2121 Quantum Universe 390833306, by Deutsche Forschungsgemeinschaft through a German-Israeli Project Cooperation (DIP) grant “Holography and the Swampland” and by Deutsche Forschungsgemeinschaft through the Collaborative Research Center 1624 “Higher Structures, Moduli Spaces and Integrability.” 

\newpage
\appendix

\section{Additional background and proof of Lemma \ref{lem:weights}}
\label{app-lemma}

In this appendix, we prove Lemma~\ref{lem:weights}. The main additional ingredients we will need are the properties of Deligne splittings and the commuting $\mathfrak{sl}(2)$-algebras that are provided by the $\mathrm{Sl}(2)$-orbit theorem. We first recall the necessary background in Section~\ref{subsec:preliminaries}. The proof of Lemma~\ref{lem:weights} is then given in Section~\ref{subsec:lem-proof}. 

\subsection{Additional background}
\label{subsec:preliminaries}
We (mostly) follow the notation of the physics-oriented reviews \cite{vandeHeisteeg:2022gsp,Monnee:2024gsq}. For further mathematical details and proofs we refer the reader to \cite{schmid,CKS,CattaniKaplan1989}. 

\subsubsection*{Limiting mixed Hodge structures}
Let $F_{(n)}$ denote the limiting Hodge filtration at the boundary point $p_\infty$, arising in the nilpotent orbit expansion of the Hodge filtration
\begin{equation}
    F^p \approx F^p_{\mathrm{nil}}= \mathrm{exp}\left(\sum_{i=1}^n t^i N_i\right)F_{(n)}^p\,.
\end{equation}
In particular, $F_{(n)}^D=\mathbb C a_0$. The nilpotent orbit theorem and its mixed-Hodge-theoretic consequences imply that
\begin{equation}
    \bigl(W(N_{(n)}),F_{(n)}\bigr)\,,
    \qquad N_{(n)}=\sum_{i=1}^nN_i\,,
\end{equation}
is a polarized mixed Hodge structure. Note that in the main text $N_{(n)}$ was denoted as $N$. 

More generally, a mixed Hodge structure $(W,F)$ admits a canonical decomposition
\begin{equation}
    H_{\mathbb C}=\bigoplus_{p,q}I^{p,q}\,,
    \qquad
    F^p=\bigoplus_{r\geq p}I^{r,s}\,,
    \qquad
    W_\ell=\bigoplus_{r+s\leq\ell}I^{r,s}\,,
    \label{eq:app-deligne-filtrations}
\end{equation}
called the \emph{Deligne splitting} $I^{r,s}$. In general, the spaces $I^{p,q}$ satisfy the following property under complex conjugation:
\begin{equation}
    \overline{I^{p,q}}=I^{q,p}
    \quad\bmod\bigoplus_{\substack{r<q\\ s<p}}I^{r,s}\,.
    \label{eq:app-deligne-conjugation}
\end{equation}
In words, $I^{p,q}$ is related to $\overline{I^{q,p}}$ up to terms that are positioned lower in the splitting. For the limiting mixed Hodge structures considered here, effectivity implies that $I^{p,q}=0$ unless $0\leq p,q\leq D$. 

We will denote the Deligne splitting of $(W(N_{(n)}),F_{(n)})$ by $I^{p,q}_{(n)}$. Of crucial importance is the action of the log-monodromy operators $N_i$ on the Deligne splitting:
\begin{equation}
    N_i I^{p,q}_{(n)}\subseteq I^{p-1,q-1}_{(n)}\,.
    \label{eq:app-monodromy-bidegree}
\end{equation}

\subsubsection*{Real splittings}

A mixed Hodge structure is called \emph{$\mathbb R$-split} if the conjugation relation is exact, namely $\overline{I^{p,q}}=I^{q,p}$. In this case, the Deligne splitting takes the simpler form
\begin{equation}
    I^{p,q}=F^p\cap\overline{F^q}\cap W_{p+q}\,.
    \label{eq:app-real-deligne}
\end{equation}
For a general mixed Hodge structure, a theorem of Deligne states that there exists a real nilpotent operator $\delta$ such that
\begin{equation}
    \widetilde F=e^{-\mathrm i\delta}F
\end{equation}
defines an $\mathbb R$-split mixed Hodge structure with the same weight filtration. The operator $\delta$ satisfies two important properties:
\begin{equation}\label{eq:prop-delta}
    [\delta, N_i] = 0\,,\qquad \delta(I^{p,q})\subseteq \bigoplus_{r<p, s<q} I^{r,s}\,.
\end{equation}
The first property states that $\delta$ commutes with all log-monodromy operators. The second property implies that $\delta$ always lowers both Deligne indices $(p,q)$ by at least one. 

The splitting used in the several-variable $\mathrm{Sl}(2)$ theorem involves a further real operator $\zeta$ and is given by\footnote{The notation for the operator $\zeta$ should not be confused with the spectator coordinates $\zeta^\alpha$.}
\begin{equation}
    \hat F =e^\zeta e^{-\mathrm i\delta}F\,.
    \label{eq:app-splitting-map}
\end{equation}
Here, $\zeta$ is a universal Lie polynomial in the components of $\delta$ with respect to the intermediate Deligne splitting defined by $\widetilde F$, see for example \cite{vandeHeisteeg:2022gsp,Monnee:2024gsq} for explicit expressions for the cases $D\leq 4$. In particular, it satisfies the same properties \eqref{eq:prop-delta}. The corresponding Deligne splitting will be denoted by
\begin{equation}
    \hat I^{p,q}=e^\zeta e^{-\mathrm i\delta}I^{p,q}\,,
    \label{eq:app-split-bigrading}
\end{equation}
and will be referred to as being \textit{sl(2)-split}.

\subsubsection*{A web of limiting mixed Hodge structures}
So far, we have discussed the properties of the limiting mixed Hodge structure associated with the sum of all monodromies $N_{(n)}$, which captured the features of sending all saxions $s^1,\ldots, s^n\to\infty$. In the following, we recall how to construct all other limiting mixed Hodge structures. 

Fix a permutation $\sigma$. For notational simplicity, we relabel the monodromy logarithms according to this ordering and write
\begin{equation}
    N_{(k)}=N_1+\cdots+N_k\,,\qquad
    W^{(k)}=W(N_{(k)})\,,\qquad
    d_k=\max\{r\in\mathbb Z_{\geq0}:N_{(k)}^r a_0\neq0\}\,.
    \label{eq:app-ordered-data}
\end{equation}
Thus $N_{(k)}$ and $d_k$ stand for $N_k^\sigma$ and $d_k^\sigma$ in the main text, while $N_{(n)}=N$ and $d_n=d$. The compatible splittings are then constructed recursively. First, let
\begin{equation}
    \hat F_{(n)}=e^{\zeta_{(n)}} e^{-\mathrm{i}\delta_{(n)}}F_{(n)}
\end{equation}
be the splitting \eqref{eq:app-splitting-map} of the full limiting mixed Hodge structure. For $k=n-1,\ldots, 1$, we recursively define
\begin{equation}\label{eq:filtration-recursion}
    F_{(k)}=e^{\mathrm{i}N_{k+1}}\hat F_{(k+1)}\,,
    \qquad
    \hat F_{(k)}
    =e^{\zeta_{(k)}}e^{-\mathrm{i}\delta_{(k)}}F_{(k)}\,.
\end{equation}
At each stage, $(W^{(k)},F_{(k)})$ is a mixed Hodge structure polarized by $N_{(k)}$, and every $N_i$ with $i\leq k$ has bidegree $(-1,-1)$. The operators $\zeta_{(k)}$ and $\delta_{(k)}$ therefore satisfy
\begin{equation}
    [\delta_{(k)},N_i]=[\zeta_{(k)},N_i]=0\,,
    \qquad i\leq k\,.
\end{equation}
The sl(2)-split mixed Hodge structure
$(W^{(k)},\hat F_{(k)})$ is also polarized by $N_{(k)}$. We denote its Deligne splitting by $\hat I^{p,q}_{(k)}$. 

\subsubsection*{Commuting $\mathfrak{sl}(2)$-algebras}
Given an sl(2)-split mixed Hodge structure $\hat{I}^{p,q}_{(k)}$ constructed above, we define a grading operator $N^0_{(k)}$ as follows
\begin{equation}
    N^0_{(k)}v=(p+q-D)v\,,
    \qquad v\in\hat I^{p,q}_{(k)}\,.
    \label{eq:app-grading-action}
\end{equation}
By definition, the eigenspaces of the grading operators $N^0_{(k)}$ of weight at most $\ell$ together recover the weight filtration $W^{(k)}_{D+\ell}$. In other words:
\begin{equation}
    W^{(k)}_{D+\ell}
    =\bigoplus_{\lambda\leq\ell}
    \ker\!\left(N^0_{(k)}-\lambda\,\mathrm{id}\right)\,.
    \label{eq:app-weight-eigenspaces}
\end{equation}
The several-variable $\mathrm{Sl}(2)$-orbit theorem associates to the above construction $n$ mutually commuting real $\mathfrak{sl}(2)$-triples $(N_i^-,N_i^+,N_i^0)$. The precise details of the construction will not be important for us, we will instead recall only the relevant properties. First, their generators satisfy the standard commutation relations
\begin{equation}
    [N_i^0,N_i^\pm]=\pm2N_i^\pm\,,
    \qquad [N_i^+,N_i^-]=N_i^0\,,
    \label{eq:app-sl2-relations}
\end{equation}
and generators belonging to different triples commute. Furthermore, all generators preserve the polarization infinitesimally. The grading operators above are related to the $N^0_i$ by
\begin{equation}
    N^0_{(k)}=\sum_{i=1}^kN_i^0\,.
    \label{eq:app-cumulative-neutral}
\end{equation}
In particular, together with the commutation relations \eqref{eq:app-sl2-relations} this implies 
\begin{equation}
    [N^0_{(k)},N^0_{(\ell)}]=0\,,
    \qquad [N^0_{(k)},N_j^-]=0\quad\text{for }j>k\,.
    \label{eq:app-cumulative-commutators}
\end{equation}
We note that the lowering operators need not coincide with the original $N_i$, but instead arise as suitable projections of $N_i$ on certain eigenspaces of the adjoint action of the grading operators $N^0_k$. This in particular guarantees the second relation in \eqref{eq:app-cumulative-commutators}. 

In terms of the lowering operators $N^-_i$, one finds the following compatibility between the various filtrations that appear in the sl(2)-split mixed Hodge structures:
\begin{equation}
    \hat F_{(k)}
    =e^{\mathrm iN_{k+1}^-}\hat F_{(k+1)}
    =\exp\!\left(\mathrm i\sum_{j=k+1}^nN_j^-\right)
      \hat F_{(n)}\,.
    \label{eq:app-compatible-filtrations}
\end{equation}
This relation will play an important role in the proof of Lemma~\ref{lem:weights}, as we will now discuss. 

\subsection{Proof of Lemma~\ref{lem:weights}}
\label{subsec:lem-proof}

For convenience, we restate the lemma in the notation of the main text.

\noindent\textbf{Lemma~\ref{lem:weights}.}
Let $q=N^d\overline{a_0}$. For every permutation $\sigma\in S_n$ and every $1\leq k\leq n$, we have
\begin{equation}
    a_0\in W(N_k^\sigma)_{D+d_k^\sigma}\,,
    \qquad
    q\in W(N_k^\sigma)_{D-d_k^\sigma}\,.
\end{equation}
Furthermore,
\begin{equation}
    N_iq=0\qquad\text{for all }i=1,\ldots,n\,.
\end{equation}

\begin{proof}
Fix an ordering $\sigma$ and use the notation \eqref{eq:app-ordered-data}. Let us first determine the position of $a_0$ and $q$ in the full limiting mixed Hodge structure. Since $F^D_{(n)}$ is one-dimensional, it equals $I_{(n)}^{D,r}$ for some $0\leq r \leq D$. The isomorphism
\begin{equation}
    N_{(n)}^r:\quad \mathrm{Gr}_{D+r}^{W^{(n)}}\longrightarrow \mathrm{Gr}_{D-r}^{W^{(n)}}
\end{equation}
then gives $N^ra_0\neq 0$. At the same time, $N^{r+1}a_0=0$ since it would have second Deligne index equal to $-1$. Hence, we conclude that $d=r$ and thus
\begin{equation}
    F_{(n)}^D=I^{D,d}_{(n)}=\mathbb C a_0\,.
    \label{eq:app-top-component}
\end{equation}
By \eqref{eq:app-deligne-conjugation}, the conjugate $\overline{a_0}$ lies in $I^{d,D}_{(n)}$ up to components in $I^{r,s}_{(n)}$ with $r<d$ and $s<D$. Applying $N^d$ annihilates all these lower components, since their first index would become negative. Consequently,
\begin{equation}
    q\in I^{0,D-d}_{(n)}\,.
    \label{eq:app-q-component}
\end{equation}
Equation~\eqref{eq:app-monodromy-bidegree} and the fact that $I^{-1,D-d-1}=\emptyset$ now imply
\begin{equation}
    \boxed{N_iq=0\,,
    \qquad i=1,\ldots,n\,.}
    \label{eq:app-q-monodromy-invariant}
\end{equation}
This proves the monodromy-invariance statement.

\subsubsection*{Position of $\hat{a}_0$}
Next, set $\hat a_0=e^{\zeta_{(n)}}e^{-i\delta_{(n)}}a_0$. Since $q$ has first Deligne index zero, the operators $\delta_{(n)}$ and $\zeta_{(n)}$ act annihilate $q$ by \eqref{eq:prop-delta}. Thus
\begin{equation}
    \hat a_0\in\hat I^{D,d}_{(n)}\,,
    \qquad q\in\hat I^{0,D-d}_{(n)}\,.
\end{equation}
Moreover, since $\delta_{(n)}$ and $\zeta_{(n)}$ preserve the polarization, we have $(q,\hat{a}_0)=(q,a_0)\neq 0$. 

To identify the weights of $\hat a_0$ under the various grading operators $N^0_{(k)}$, we recursively define (recall \eqref{eq:filtration-recursion})
\begin{equation}
    \hat a_0^{(n)}=\hat a_0\,,
    \qquad
    \hat a_0^{(k)}=e^{\zeta_{(k)}}e^{-\mathrm{i}\delta_{(k)}}e^{\mathrm iN_{k+1}}                        \hat a_0^{(k+1)}\,,
    \label{eq:app-partial-top-generators}
\end{equation}
such that $\hat{a}_0^{(k)}\in \hat F_{(k)}^D$. For a fixed $k$, every map occurring between $a_0$ and $\hat a_0^{(k)}$ is invertible and commutes with $N_{(k)}$. Therefore,
\begin{equation}
    N_{(k)}^r\hat a_0^{(k)}=0
    \quad\Longleftrightarrow\quad
    N_{(k)}^r a_0=0\,.
\end{equation}
Applying the same argument leading to \eqref{eq:app-top-component} to the mixed Hodge structure $(W^{(k)},\hat F_{(k)})$, we obtain
\begin{equation}
    \hat F^D_{(k)}=\hat I^{D,d_k}_{(k)}\,,
    \qquad
    N^0_{(k)}\hat a_0^{(k)}=d_k\hat a_0^{(k)}\,.
    \label{eq:app-partial-top-weight}
\end{equation}
We now use the $\mathfrak{sl}(2)$ commutation relations to determine the action of $N^0_{(k)}$ on $\hat a_0$. Equation~\eqref{eq:app-compatible-filtrations} and the second relation in \eqref{eq:app-cumulative-commutators} imply
\begin{align}
    N^0_{(k)}\hat F^p_{(n)}
    &=N^0_{(k)}
      \exp\!\left(-\mathrm i\sum_{j=k+1}^nN_j^-\right)
      \hat F^p_{(k)}\notag\\
    &=\exp\!\left(-\mathrm i\sum_{j=k+1}^nN_j^-\right)
      N^0_{(k)}\hat F^p_{(k)}\notag\\
    &\subseteq\hat F^p_{(n)}\,.
    \label{eq:app-grading-preserves-full-filtration}
\end{align}
Setting $p=D$, it follows that
\begin{equation}
    N^0_{(k)}\hat a_0=d_k\hat a_0\,.
    \label{eq:app-full-top-partial-weight}
\end{equation}
Equation~\eqref{eq:app-weight-eigenspaces} therefore gives
\begin{equation}
    \hat a_0\in W^{(k)}_{D+d_k}\,.
\end{equation}
Finally, the transformation $e^{\zeta_{(n)}}e^{-i\delta_{(n)}}$ commutes with $N_{(k)}$, and hence preserves the weight filtration $W^{(k)}$. Applying its inverse to $\hat a_0\in W^{(k)}_{D+d_k}$ therefore gives
\begin{equation}\label{eq:pos-a0}
    \boxed{a_0\in W^{(k)}_{D+d_k}}
\end{equation}

\subsubsection*{Position of $q$}
To determine the corresponding eigenvalue of $q$, recall that
\begin{equation}
    \hat I^{p,q}_{(n)}
    =\hat F^p_{(n)}\cap
      \overline{\hat F^q_{(n)}}\cap W^{(n)}_{p+q}\,.
\end{equation}
The operator $N^0_{(k)}$ preserves the first factor by \eqref{eq:app-grading-preserves-full-filtration}, and the second because it is real. It preserves the third because it commutes with $N^0_{(n)}$, which grades $W^{(n)}$. Hence $N^0_{(k)}$ preserves every $\hat I^{p,q}_{(n)}$. In particular, it acts on the one-dimensional space $\hat I^{0,D-d}_{(n)}=\mathbb C q$ as
\begin{equation}
    N^0_{(k)}q=\lambda_k q\,,
\end{equation}
where $\lambda_k$ is some eigenvalue which we now determine. Using preservation of the polarization, we  find
\begin{align}
    0&=(N^0_{(k)}q,\hat a_0)
       +(q,N^0_{(k)}\hat a_0)\notag\\
     &=(\lambda_k+d_k)(q,\hat a_0)\,.
\end{align}
The pairing $(q,\hat{a}_0)$ is non-zero, hence $\lambda_k=-d_k$. Equation~\eqref{eq:app-weight-eigenspaces} therefore gives
\begin{equation}\label{eq:pos-q}
    \boxed{q\in W^{(k)}_{D-d_k}}
\end{equation}

\subsubsection*{Completing the proof}
Since the ordering was arbitrary, the inclusions \eqref{eq:pos-a0} and \eqref{eq:pos-q} hold for every permutation $\sigma$. Together with the monodromy-invariance \eqref{eq:app-q-monodromy-invariant}, this proves all statements of Lemma~\ref{lem:weights}.
\end{proof}

\bibliographystyle{JHEP}
\bibliography{references}

\end{document}